\documentclass[11pt,letterpaper]{article}

\usepackage[margin=0.92in]{geometry}
\usepackage[T1]{fontenc}
\usepackage{lmodern}
\usepackage{microtype}
\usepackage{amsmath,amssymb,amsthm,mathtools,bm}
\usepackage{mathrsfs}
\usepackage{booktabs,array,longtable,tabularx}
\usepackage{float}
\usepackage{enumitem}
\usepackage{xcolor}
\usepackage{graphicx}
\usepackage{tikz}
\usetikzlibrary{arrows.meta,positioning,calc,decorations.pathmorphing,fit,backgrounds}
\usepackage{hyperref}
\usepackage[nameinlink,capitalize,noabbrev]{cleveref}
\usepackage{caption}
\usepackage{subcaption}
\usepackage{url}
\usepackage{bm}

\hypersetup{
  hidelinks,
  pdftitle={Confinement as Decoding: Higher Form Codes and Lattice Yang--Mills Theory},
  pdfauthor={Ning Bao},
  pdfsubject={Quantum error correction, higher form gauge theory, and lattice Yang--Mills theory}
}

\definecolor{deepblue}{RGB}{25,55,105}
\definecolor{softblue}{RGB}{235,242,252}
\definecolor{softgray}{RGB}{245,246,248}
\definecolor{darkgray}{RGB}{65,70,78}
\definecolor{softgold}{RGB}{249,244,226}
\definecolor{softgreen}{RGB}{237,248,240}

\newtheorem{theorem}{Theorem}[section]
\newtheorem{proposition}[theorem]{Proposition}
\newtheorem{lemma}[theorem]{Lemma}
\newtheorem{corollary}[theorem]{Corollary}
\newtheorem{assumption}[theorem]{Assumption}

\theoremstyle{definition}
\newtheorem{definition}[theorem]{Definition}
\newtheorem{remark}[theorem]{Remark}
\newtheorem{example}[theorem]{Example}

\crefname{theorem}{Theorem}{Theorems}
\Crefname{theorem}{Theorem}{Theorems}
\crefname{proposition}{Proposition}{Propositions}
\Crefname{proposition}{Proposition}{Propositions}
\crefname{lemma}{Lemma}{Lemmas}
\Crefname{lemma}{Lemma}{Lemmas}
\crefname{corollary}{Corollary}{Corollaries}
\Crefname{corollary}{Corollary}{Corollaries}
\crefname{assumption}{Assumption}{Assumptions}
\Crefname{assumption}{Assumption}{Assumptions}
\crefname{definition}{Definition}{Definitions}
\Crefname{definition}{Definition}{Definitions}
\crefname{remark}{Remark}{Remarks}
\Crefname{remark}{Remark}{Remarks}
\crefname{example}{Example}{Examples}
\Crefname{example}{Example}{Examples}

\numberwithin{equation}{section}
\allowdisplaybreaks

\newcommand{\ZN}{\mathbb{Z}_{N}}
\newcommand{\Ztwo}{\mathbb{Z}_{2}}
\newcommand{\ZZ}{\mathbb{Z}}
\newcommand{\RR}{\mathbb{R}}

\newcommand{\SU}{\mathrm{SU}}
\newcommand{\SO}{\mathrm{SO}}
\newcommand{\PSU}{\mathrm{PSU}}
\newcommand{\im}{\operatorname{im}}
\newcommand{\supp}{\operatorname{supp}}
\newcommand{\Tr}{\operatorname{Tr}}
\newcommand{\tr}{\operatorname{tr}}
\newcommand{\Prob}{\mathbb{P}}
\newcommand{\E}{\mathbb{E}}
\newcommand{\cH}{\mathcal{H}}
\newcommand{\cE}{\mathcal{E}}
\newcommand{\cZ}{\mathcal{Z}}
\newcommand{\cC}{\mathcal{C}}
\newcommand{\cN}{\mathcal{N}}
\newcommand{\cL}{\mathcal{L}}
\newcommand{\cM}{\mathcal{M}}

\newcommand{\cW}{\mathcal{W}}
\newcommand{\cF}{\mathcal{F}}
\newcommand{\cR}{\mathcal{R}}
\newcommand{\dd}{\mathrm{d}}
\newcommand{\e}{\mathrm{e}}
\newcommand{\ii}{\mathrm{i}}
\newcommand{\ket}[1]{\lvert #1\rangle}
\newcommand{\bra}[1]{\langle #1\rvert}
\newcommand{\braket}[2]{\langle #1\mid #2\rangle}
\newcommand{\abs}[1]{\left\lvert #1\right\rvert}
\newcommand{\norm}[1]{\left\lVert #1\right\rVert}
\newcommand{\set}[1]{\left\{#1\right\}}

\newcommand{\YM}{\mathrm{YM}}
\newcommand{\ML}{\mathrm{ML}}

\newcommand{\wideG}{\widehat G}

\newcommand{\Tcode}{\mathsf{T}_{\rm code}}

\setlist[itemize]{leftmargin=1.5em,itemsep=0.18em,topsep=0.3em}
\setlist[enumerate]{leftmargin=1.7em,itemsep=0.22em,topsep=0.3em}

\begin{document}

\thispagestyle{empty}
\begin{center}
{\LARGE\bfseries Confinement as Decoding:\par}
\vspace{0.12cm}
{\LARGE\bfseries Higher Form Codes and Lattice Yang--Mills Theory\par}
\vspace{0.75cm}
{\large Ning Bao\par}
\vspace{0.20cm}
{\small Computational Science Initiative, Brookhaven National Laboratory, Upton, NY 11973, USA\\
Department of Physics, Northeastern University, Boston, MA 02115, USA\\
\texttt{ningbao75@gmail.com}\par}
\end{center}
\vspace{0.25cm}

\begin{abstract}
We study the relationship between quantum error correction, confinement, and lattice Yang--Mills theory. We first formulate decoding for finite Abelian homological codes in terms of higher form gauge fields. For positive local noise, the logical classes are topological sectors of a Nishimori ensemble, and the optimal decoding error is determined by the relative weights of the nontrivial sectors. We derive Fourier relations between logical probabilities, disorder operators, and information in the channel environment, and we give contour and fractional moment criteria for a threshold. We then study a four dimensional $\ZN$ memory and its possible relation to confining $\PSU(N)$ vacua. Finally, we define a finite curvature center sheet model coupled to Wilson $\SU(N)$ link variables. In this model the conditional logical probabilities are center twisted Yang--Mills partition functions. A strong coupling expansion gives the leading effective interaction for the syndrome and shows that local syndrome correlations can decay even when the global sheet sectors are mixed. We also show that the likelihood for a separated pair of syndrome worldlines is the center monopole correlator. Its decay determines a transfer matrix mass. This distinguishes the suppression of global flux sectors from the local spectral information needed to discuss a mass gap.
\end{abstract}

\clearpage
\tableofcontents
\clearpage

\section{Introduction}
\label{sec:intro}

Quantum error correction has become a useful language for several problems in high energy physics. In holography, the bulk to boundary map can be viewed as a redundant encoding. In topological phases, logical operators are directly related to extended excitations. In lattice gauge theory, the same type of sector structure appears through electric and magnetic fluxes. Dennis, Kitaev, Landahl, and Preskill related decoding of the surface code to a disordered statistical model on the Nishimori line \cite{Nishimori}, and Wang, Harrington, and Preskill showed that repeated syndrome measurements give a random plaquette gauge theory \cite{DKLP,WHP}. Chubb and Flammia later gave a general statistical mechanical construction for stabilizer and subsystem codes with correlated Pauli noise \cite{ChubbFlammia}. Another approach begins with superselection sectors. In that setting the sector structure can imply the Knill--Laflamme condition, with proton and neutron sectors in quantum chromodynamics providing one example \cite{BaoSuperselection}. Gauge redundancy and Gauss law constraints have also been used to protect quantum simulations of lattice gauge theories \cite{RajputRoggeroWiebe,CarenaEtAl,SpagnoliEtAl,LacambraEtAl,YaoSU2}, and quantum reference frame methods give a related connection between gauge redundancy and correctability \cite{CarrozzaEtAl}.

We ask whether these observations can be organized into a general relation between confinement and decoding. We study one Pauli shift error sector of a finite Abelian homological CSS code with ideal syndrome information. The measured syndrome fixes an affine space of errors. The remaining ambiguity is a homology class, and the conditional probability of each class is an orbit sum. Maximum likelihood decoding therefore compares topological sectors. This finite volume statement is exact. A circuit level threshold analysis would also require the complementary Pauli sector, noisy syndrome measurements, an explicit recovery procedure, and control of coherent errors.

The sector formulation does not require a conventional local order parameter. It applies to correlated local noise, growing logical groups, partial mixing of logical sectors, and matter that screens a bare Wilson loop. For product noise, finite group Fourier duality relates the logical probabilities to Wilson and disorder amplitudes. The same transform appears in a purification of the channel and relates uncertainty at the receiver to distinguishability in the environment. We express the coherent information in terms of the syndrome and logical entropies. We also prove a contour bound that gives a nonzero threshold region for bounded geometry code families. The first moment of a logical odds ratio is exactly one on the Nishimori ensemble and therefore carries no threshold information. Fractional moments are nontrivial. At exponent one half, the relevant quantity is the Bhattacharyya affinity between the sector distribution and its logical translate.

Three objects will appear below. The first is the quantum system whose ground space stores information. The second is the classical decoding ensemble fixed by the noise distribution and the measured syndrome. The third is the mixed quantum state produced by the channel. These objects need not be governed by the same coupling. We also use two different codes. A spatial four dimensional $\ZN$ toric code is used to discuss confining $\PSU(N)$ vacua. A Euclidean two form code on a four dimensional lattice describes center vortex sheets in an $\SU(N)$ path integral. The two constructions involve related topological sectors, but their microscopic degrees of freedom are different.

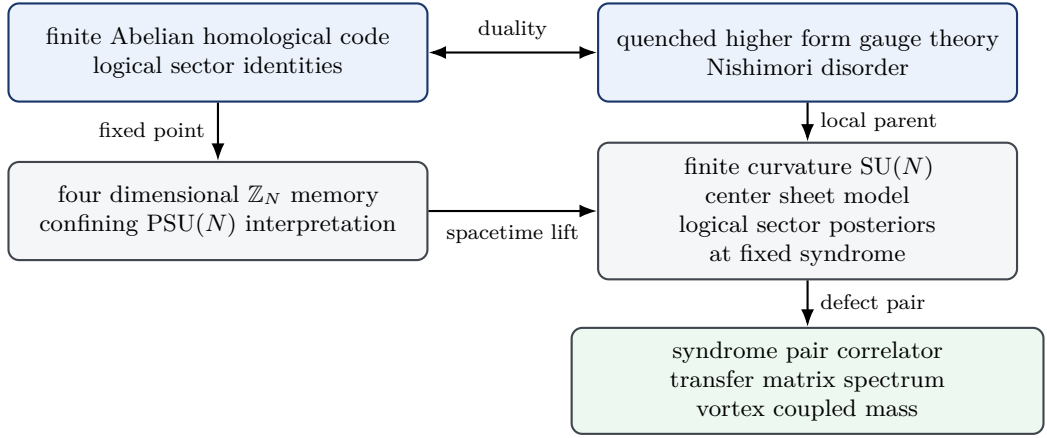
\begin{figure}[t]
\centering
\begin{tikzpicture}[>=Latex,thick]
\tikzset{
 layer/.style={draw=deepblue,rounded corners,fill=softblue,
   text width=5.05cm,minimum height=1.30cm,align=center,font=\footnotesize,
   inner xsep=7pt,inner ysep=5pt},
 app/.style={draw=darkgray,rounded corners,fill=softgray,
   text width=5.05cm,minimum height=1.30cm,align=center,font=\footnotesize,
   inner xsep=7pt,inner ysep=5pt},
 spectrum/.style={draw=darkgray,rounded corners,fill=softgreen,
   text width=5.75cm,minimum height=1.30cm,align=center,font=\footnotesize,
   inner xsep=7pt,inner ysep=5pt},
 lab/.style={font=\scriptsize,align=center,fill=white,inner sep=1.5pt}
}
\node[layer] (engine) at (-3.90,1.80)
  {finite Abelian homological code\\logical sector identities};
\node[layer] (sm) at (3.90,1.80)
  {quenched higher form gauge theory\\Nishimori disorder};
\node[app] (spatial) at (-3.90,-0.30)
  {four dimensional $\ZN$ memory\\confining $\PSU(N)$ interpretation};
\node[app] (euclid) at (3.90,-0.30)
  {finite curvature $\SU(N)$ center sheet model\\logical sector posteriors\\at fixed syndrome};
\node[spectrum] (mass) at (3.90,-2.55)
  {syndrome pair correlator\\transfer matrix spectrum\\vortex coupled mass};

\draw[<->] (engine.east) --
  node[midway,above=3pt,lab]{duality} (sm.west);
\draw[->] (engine.south) --
  node[midway,left=3pt,lab]{fixed point} (spatial.north);
\draw[->] (sm.south) --
  node[midway,right=3pt,lab]{local parent} (euclid.north);
\draw[->] (spatial.east) --
  node[midway,below=3pt,lab]{spacetime lift} (euclid.west);
\draw[->] (euclid.south) --
  node[midway,right=3pt,lab]{defect pair} (mass.north);
\end{tikzpicture}
\caption{The constructions used in this paper. The homological relation is algebraic. The $\PSU(N)$ interpretation requires additional assumptions about the infrared theory, while the center sheet posterior is an exact finite volume identity.}
\label{fig:paper_structure}
\end{figure}

We first study the four dimensional $\ZN$ toric code, which is a higher dimensional version of Kitaev's construction \cite{Kitaev}. Its logical operators, distance, excitations, and energy barriers are explicit. Gauging the electric center symmetry of $\SU(N)$ suggests that a confining adjoint theory can inherit the magnetic logical algebra of this fixed point. This interpretation requires a gapped Hamiltonian at fixed lattice spacing, a controlled reduction to the center sector, and stability of the resulting topological phase. None of these assumptions enters the Euclidean Yang--Mills construction.

For the Euclidean application, we define a correlated Pauli ensemble whose local parent theory contains Wilson $\SU(N)$ links and a $\ZN$ plaquette field. The boundary of the plaquette field is the syndrome, and its cohomology class is the logical center sheet sector. We use the model as an inference problem rather than as a microscopic model of real time Yang--Mills noise. One form gauge invariance assigns the same Yang--Mills weight to sheets that differ by a stabilizer. The conditional probability of a logical completion is therefore a center twisted partition function. At strong coupling, integrating out the links produces a local syndrome action whose first contribution is the elementary cube. A polymer can distinguish global logical sectors only by wrapping a nontrivial cycle, so the distinction is suppressed by the systolic area. Local syndrome correlations can consequently decay even when the global sheet sectors are mixed.

Nonzero syndromes also probe the spectrum. Two syndrome worldlines joined by a center sheet give the lattice center monopole insertion. After removing the known endpoint fugacity, the pair likelihood is the twisted plaquette monopole correlator. Reflection positivity gives a transfer matrix representation, and the connected decay measures the lowest mass that couples to the vortex source. This also separates confinement from a neutral mass gap. Global logical probabilities constrain sectors with center flux, while glueballs are center neutral. A full transfer matrix gap requires exponential estimates for a family of syndrome sources that has overlap with the full physical spectrum.

The transfer matrix statements first hold at fixed lattice spacing. A thermodynamic result requires estimates uniform in transverse volume and Euclidean time extent. A continuum result further requires a lower bound that remains positive in physical units along a scaling trajectory. The strong coupling expansion controls one region of the finite curvature model but does not determine the phase boundary in the continuum regime. These questions are dynamical and are not fixed by the finite volume sector identities.

Several earlier results provide useful context. Bao, Cao, and Zhu related deconfinement to error thresholds in holography \cite{BaoHolography}. Bao, Cao, Chatwin-Davies, Cheng, and Zhu showed that exact superselection can imply quantum error correction and discussed a possible connection between confinement and mass generation \cite{BaoSuperselection}. In the present setting, the sectors have finite relative weights after the syndrome is measured, so exact superselection is replaced by Bayesian inference. Li, O'Dea, and Khemani distinguish defects that diagnose logical stability from defects that probe local excitation gaps \cite{LiODeaKhemani}. Liu, Xu, Pollmann, and Knap describe decodability as an information theoretic test of emergent one form symmetry \cite{LiuXuPollmannKnap}. We combine these ideas with finite volume center backgrounds, thermal 't Hooft flux sectors, and the center monopole transfer matrix correlator \cite{tHooftFlux,GKSW,AbeEtAl,deForcrandSmekal,deForcrandNoth,MonopolePairs}.

The paper is organized as follows. Part I develops the finite Abelian decoding identities, the threshold criteria, and the mixed state diagnostics. Part II studies the four dimensional $\ZN$ memory, extensions to composite coefficient groups, and the possible relation to confining $\PSU(N)$ vacua. Part III defines the finite curvature $\SU(N)$ center sheet model, derives its logical posterior, and studies its strong coupling and thermal limits. Part IV separates information about charged flux sectors from neutral spectroscopy. It then relates syndrome pair likelihoods to center monopole correlators and gives conditions under which these correlators imply a transfer matrix or continuum gap.

\part{Finite Abelian decoding and higher form gauge theory}

\section{Homological codes and logical sector posteriors}
\label{sec:general_code}

\subsection{Chain complex and generalized Pauli code}

Let $X$ be a finite cell complex of dimension $D$. We take a finite Abelian group $G$, written additively, and choose a cell degree $k$. The relevant part of the chain complex is
\[
 C_{k+1}(X;G)\xrightarrow{\partial_{k+1}}
 C_k(X;G)\xrightarrow{\partial_k}
 C_{k-1}(X;G),
 \qquad \partial_k\partial_{k+1}=0.
\]
We place a qudit of dimension $\abs{G}$ on each $k$ cell. The computational basis is indexed by $C_k(X;G)$. For $a\in C_k(X;G)$ and $\chi\in C^k(X;\wideG)$, where $\wideG=\operatorname{Hom}(G,U(1))$, define
\[
 X_a\ket{c}=\ket{c+a},
 \qquad
 Z_\chi\ket{c}=\chi(c)\ket{c},
 \qquad
 Z_\chi X_a=\chi(a)X_aZ_\chi.
\]
The homological CSS code is stabilized by $X_{\partial b}$ for $b\in C_{k+1}(X;G)$ and by $Z_{\delta\beta}$ for $\beta\in C^{k-1}(X;\wideG)$. The shift logical group is
\[
 H:=H_k(X;G)=\ker\partial_k/\im\partial_{k+1}.
\]
The conjugate logical group is $H^k(X;\wideG)$.

\begin{lemma}[Perfect logical pairing]
\label[lemma]{lem:pontryagin}
We have a canonical isomorphism
\[
 H^k(X;\wideG)\cong \widehat{H_k(X;G)},
\]
and the pairing $\langle[\chi],[c]\rangle=\chi(c)$ is perfect. In particular, the code dimension is $\abs{H_k(X;G)}$. No freeness or coprimality assumption is needed.
\end{lemma}

\begin{proof}
Since $U(1)$ is divisible, it is injective as a $\ZZ$ module. The functor $\operatorname{Hom}(-,U(1))$ is therefore exact and commutes with homology. This gives
\[
 H^k(X;\wideG)=\operatorname{Hom}(H_k(X;G),U(1)).
\]
A stabilizer basis state is a uniform superposition over a boundary coset. The phase checks require the coset label to be a cycle. The code basis is therefore indexed by $\ker\partial_k/\im\partial_{k+1}$.
\end{proof}

The construction also applies to relative chain complexes, which describe boundaries, defects, lattice surgery, and spacetime cobordisms.

\subsection{Positive local noise and sector weights}

Let $e\in C_k(X;G)$ denote an error. We allow an arbitrary strictly positive local distribution
\begin{equation}
 Q_\theta(e)=\frac{1}{\Xi_\theta}\e^{-S_\theta(e)},
 \qquad
 S_\theta(e)=\sum_{R:\,\operatorname{diam}R\le R_0}U_{R,\theta}(e|_R).
\label{eq:local_noise}
\end{equation}
Independent noise on the cells is the special case $Q(e)=\prod_xq_x(e_x)$. Measurement of the conjugate stabilizers gives the syndrome
\[
 \sigma(e)=\partial e.
\]
For each syndrome, choose $e_\sigma$ with $\partial e_\sigma=\sigma$. We also choose a cycle representative $\gamma_h$ for every $h\in H$. The unnormalized weight of the logical sector $h$ is
\begin{equation}
 Z_{\sigma,h}=\sum_{b\in B_k}Q_\theta(e_\sigma+\gamma_h+b),
 \qquad B_k=\im\partial_{k+1}.
\label{eq:general_sector}
\end{equation}
Then
\begin{equation}
 \Prob(h\mid\sigma)=\frac{Z_{\sigma,h}}{\sum_{j\in H}Z_{\sigma,j}}.
\label{eq:posterior_general}
\end{equation}
A different choice of $e_\sigma$ or of the representatives $\gamma_h$ only relabels the sectors.

\begin{theorem}[Higher form Nishimori identity]
\label[theorem]{thm:nishimori_general}
For every syndrome and logical class,
\begin{equation}
 Z_{\sigma,h}=\frac{1}{\abs{\ker\partial_{k+1}}\,\Xi_\theta}
 \sum_{a\in C_{k+1}(X;G)}
 \exp\!\left[-S_\theta(e_\sigma+\gamma_h+\partial a)\right].
\label{eq:gauge_partition_general}
\end{equation}
Thus maximum likelihood decoding compares topological sectors of a quenched chain gauge model. Suppose that $X$ is a cellulation of an oriented $D$ manifold with a cellular dual. Then $a$ is a gauge field of form degree $D-k-1$ on the dual complex, while $e_\sigma+\gamma_h$ is the quenched background field strength. On the branch containing the actual error, the background is drawn with the same Boltzmann weight that appears in the partition sum. The generalized Nishimori condition is therefore automatic.
\end{theorem}

\begin{proof}
The map $\partial_{k+1}:C_{k+1}\to B_k$ is onto, and every fiber has size $\abs{\ker\partial_{k+1}}$. We can therefore replace the sum over $b\in B_k$ by a sum over all of its preimages, which gives \eqref{eq:gauge_partition_general}. If we choose the actual error as the reference, the quenched background is distributed according to $Q_\theta$. This is the same local action that appears in the gauge partition function.
\end{proof}

For $G=\Ztwo$, $D=2$, and $k=1$, the dual model is the random bond Ising model. Repeated syndrome extraction replaces the spatial complex by a spacetime complex and gives the random plaquette gauge model of \cite{WHP}. The finite volume identities are unchanged by this lift. A threshold bound must instead use the incidence degree, cell count, and distance of a nontrivial cycle in the spacetime complex.

\subsection{Logical defect fugacity}

For each syndrome, choose a class $h_*(\sigma)$ of maximum weight. We define
\begin{equation}
 \Delta_\sigma(g)=
 \log\frac{Z_{\sigma,h_*}}{Z_{\sigma,h_*+g}}\ge0,
 \qquad
 R_\sigma=\sum_{g\ne0}\e^{-\Delta_\sigma(g)}.
\label{eq:defect_fugacity_general}
\end{equation}

\begin{theorem}[Optimal failure probability]
\label[theorem]{thm:exact_failure}
The conditional and averaged optimal failure probabilities are
\begin{equation}
 p_{\rm fail}^{\ML}(\sigma)=\frac{R_\sigma}{1+R_\sigma},
 \qquad
 p_{\rm fail}^{\ML}=\E_\sigma\!\left[\frac{R_\sigma}{1+R_\sigma}\right].
\label{eq:exact_failure}
\end{equation}
Reliable decoding is therefore equivalent to $R_\sigma\to0$ in probability. If $\abs H$ is fixed, this is equivalent to
\[
 \min_{g\ne0}\Delta_\sigma(g)\longrightarrow+\infty
\]
in probability. For growing $H$,
\begin{equation}
 \e^{-\Delta_{\min,\sigma}}
 \le R_\sigma\le(\abs H-1)\e^{-\Delta_{\min,\sigma}},
\label{eq:defect_entropy}
\end{equation}
so a sufficient condition is $\Delta_{\min,\sigma}-\log(\abs H-1)\to+\infty$.
\end{theorem}

\begin{proof}
Divide the normalization of the posterior by $Z_{\sigma,h_*}$. Translation by $h_*$ permutes the elements of $H$, so the denominator becomes $1+R_\sigma$. Since $0\le R/(1+R)\le1$, its expectation tends to zero exactly when the random variable tends to zero in probability. The other statements follow from \eqref{eq:defect_entropy}.
\end{proof}

A decoding threshold need not coincide with a singularity of the bulk free energy density. The condition for decoding is instead the vanishing of the total fugacity of nontrivial topological twists.

\section{Fourier duality and mixed state diagnostics}
\label{sec:duality}

\subsection{Finite group Fourier duality}

We now take independent noise,
\[
 Q(e)=\prod_xq_x(e_x),
\]
and use the Fourier convention
\begin{equation}
 q_x(g)=\sum_{\chi\in\wideG}\widehat q_x(\chi)\chi(g),
 \qquad
 \widehat q_x(\chi)=\frac{1}{\abs G}\sum_{g\in G}q_x(g)\overline{\chi(g)}.
\label{eq:fourier_q}
\end{equation}
Let $Z^k(X;\wideG)=\ker\delta_k$ and set
\[
 \cW(\chi)=\prod_x\widehat q_x(\chi_x).
\]
Character orthogonality gives
\begin{equation}
 Z_0(e)=\abs{B_k}
 \sum_{\chi\in Z^k(X;\wideG)}\cW(\chi)\chi(e).
\label{eq:dual_partition}
\end{equation}
This identity is algebraic and remains valid when the Fourier weights are signed or complex. If the weights define a positive dual model, the ratios below are ordinary expectation values. Otherwise, they are normalized disorder amplitudes.

\begin{theorem}[Logical odds are dual Wilson amplitudes]
\label[theorem]{thm:wilson_odds}
Choose the actual error $e$ as the reference. Its logical class is then $0$. For any $h\in H$,
\begin{equation}
 \frac{\Prob(h\mid\partial e)}{\Prob(0\mid\partial e)}
 =\frac{Z_0(e+\gamma_h)}{Z_0(e)}
 =\frac{\sum_{\chi\in Z^k}\cW(\chi)\chi(e)\chi(\gamma_h)}
 {\sum_{\chi\in Z^k}\cW(\chi)\chi(e)}.
\label{eq:wilson_odds}
\end{equation}
The factor $\chi(\gamma_h)=\langle[\chi],h\rangle$ is a Wilson operator in the dual theory. In the original decoding theory, it is a disorder or 't Hooft insertion.
\end{theorem}

\begin{proof}
The first equality follows from the shift relation in \eqref{eq:general_sector}. Applying \eqref{eq:dual_partition} to the numerator and denominator gives the second equality. A cocycle vanishes on boundaries, so $\chi(\gamma_h)$ depends only on $h$.
\end{proof}

For the surface code, \eqref{eq:wilson_odds} is the Kadanoff--Ceva disorder line or domain wall free energy \cite{KramersWannier,Wegner,KadanoffCeva}. In higher dimensions, it becomes a wrapped Wilson surface or higher form disorder operator. The quantity that enters decoding is the posterior logical tension
\begin{equation}
 \tau_h=-\limsup_{L\to\infty}\frac{1}{d_h(L)}
 \log\frac{Z_0(e+\gamma_h)}{Z_0(e)},
 \qquad
 d_h(L)=\min_{[c]=h}\abs{\supp c},
\label{eq:logical_tension}
\end{equation}
where the limit can be taken in probability or after choosing a disorder average.

\subsection{Fourier relation between the receiver and the environment}

Purify the Pauli channel by
\begin{equation}
 V\ket{\psi}=\sum_e\sqrt{Q(e)}\,X_e\ket{\psi}\otimes\ket{e}_E.
\label{eq:noise_purification}
\end{equation}
Let $\ket{w}$ be an eigenstate of the shift logical algebra, with $w\in\widehat H$. Two errors with the same syndrome differ by a logical cycle. They therefore give the same receiver state, up to the character associated with that cycle.

\begin{proposition}[Environment Gram matrix]
\label[proposition]{prop:environment}
Conditioned on a syndrome $\sigma$, the environment state for input $\ket{w}$ is pure within that syndrome block. For two logical eigenstates $w,w'\in\widehat H$,
\begin{equation}
 \braket{\Phi_\sigma^{w'}}{\Phi_\sigma^{w}}
 =\sum_{h\in H}\Prob(h\mid\sigma)\,\langle w-w',h\rangle.
\label{eq:environment_fourier}
\end{equation}
The Gram matrix of the environment is therefore the finite group Fourier transform of the logical posterior at the receiver. If the posterior concentrates on one class, the environment states are parallel up to phases. If the posterior is uniform, states with different conjugate labels are orthogonal.
\end{proposition}

\begin{proof}
Fix a syndrome and choose representatives $e_{\sigma,h}$. The receiver states obey
\[
 X_{e_{\sigma,h}}\ket{w}
 =\langle w,h\rangle\,X_{e_{\sigma,0}}\ket{w}
\]
up to stabilizers. All receiver vectors in the syndrome block are therefore collinear, and the corresponding environment block is pure. Taking the overlap for inputs $w$ and $w'$ gives \eqref{eq:environment_fourier}.
\end{proof}

Equation \eqref{eq:environment_fourier} has a direct operational interpretation. Logical uncertainty at the receiver is the Fourier dual of distinguishability between conjugate sectors in the environment.

\subsection{Coherent information}

Let $R$ be a reference system maximally entangled with the code space, and send the other half through the shift noise channel.

\begin{lemma}[Coherent information]
\label[lemma]{lem:coherent_information}
The coherent information is
\begin{equation}
 I_c(R\rangle B)=\log\abs H-\E_\sigma H\bigl(\Prob(\cdot\mid\sigma)\bigr),
\label{eq:coherent_information}
\end{equation}
where $H$ is Shannon entropy with natural logarithms. Moreover,
\begin{equation}
 p_{\rm fail}^{\ML}\le
 \frac{\E_\sigma H(\Prob(\cdot\mid\sigma))}{\log2}.
\label{eq:failure_entropy}
\end{equation}
Hence $I_c\to\log\abs H$ implies reliable decoding. Conversely, reliable decoding implies $I_c\to\log\abs H$ whenever
\begin{equation}
 p_{\rm fail}^{\ML}\log\abs H\longrightarrow0,
\label{eq:fano_condition}
\end{equation}
which is automatic for bounded $\abs H$.
\end{lemma}

\begin{proof}
The joint state of the reference and receiver is a direct sum of orthogonal syndrome and logical class blocks with probabilities $\Prob(\sigma)\Prob(h\mid\sigma)$. Tracing out the reference removes the logical label inside each syndrome block. Therefore
\[
 S(RB)=H(\sigma)+\E_\sigma H(\Prob(\cdot\mid\sigma)),
 \qquad
 S(B)=H(\sigma)+\log\abs H.
\]
This gives \eqref{eq:coherent_information}. The pointwise inequality $1-p_{\max}\le H(p)/\log2$ gives \eqref{eq:failure_entropy}. The converse is Fano's inequality together with \eqref{eq:fano_condition}.
\end{proof}

\subsection{Second replica observables and the species rule}

The Nishimori ensemble used for optimal decoding is different from the replica coupling that appears in intrinsic mixed state diagnostics. Let
\[
 \rho_0=\frac{P_{\rm code}}{\Tr P_{\rm code}},
 \qquad
 \rho_Q=\cN_Q(\rho_0),
 \qquad
 F_O=\frac{\Tr(\rho_QO\rho_QO^\dagger)}{\Tr\rho_Q^2}.
\]
For product shift noise define
\begin{equation}
 r_x(g)=\sum_{a\in G}q_x(a)q_x(a-g),
 \qquad
 R(c)=\prod_xr_x(c_x).
\label{eq:replica_kernel}
\end{equation}

\begin{proposition}[Replica two identities]
\label[proposition]{prop:replica_two}
For the shift noise channel:
\begin{enumerate}
\item If $\chi$ is a closed conjugate species operator, then $F_{Z_\chi}=1$ exactly.
\item If $\chi$ is open, then $F_{Z_\chi}=0$ exactly because it maps the state to an orthogonal stabilizer syndrome sector.
\item For a same species shift $X_a$,
\begin{equation}
 F_{X_a}=\frac{\sum_{z\in Z_k(X;G)}R(a+z)}{\sum_{z\in Z_k(X;G)}R(z)}.
\label{eq:same_species_replica}
\end{equation}
In particular, $F_{X_a}=1$ for every closed $a\in Z_k$.
\end{enumerate}
The nontrivial replica diagnostic is therefore an open defect pair in the same species. A closed logical loop is identically one, while an open operator in the conjugate species is identically zero.
\end{proposition}

\begin{proof}
A closed $Z_\chi$ commutes with the code projector. It changes each error amplitude by a phase, but the phase cancels between the two density matrices. An open $Z_\chi$ changes an $X$ stabilizer eigenvalue, so the two states have orthogonal support. For $X_a$, expand the two copies of the channel. The trace with the code projector vanishes unless $a+e'-e$ is a cycle. Summing over the difference distribution gives \eqref{eq:same_species_replica}.
\end{proof}

\section{Threshold criteria and partial condensation}
\label{sec:thresholds}

\subsection{Peierls threshold bound}

Assume independent noise and let $\Delta_X$ be the maximal degree of the incidence graph on $k$ cells, in which two cells are adjacent when they share an incident cell of degree $k-1$ or $k+1$. Define
\begin{equation}
 w(q)=\max_x\max_{g\ne0}\sum_{a\in G}\sqrt{q_x(a)q_x(a+g)},
 \qquad
 \kappa_0=e\Delta_X(\abs G-1).
\label{eq:bhattacharyya_w}
\end{equation}
Let $n_k$ be the number of $k$ cells and
\[
 d_{\min}=\min\set{\abs{\supp c}:c\in Z_k\setminus B_k}.
\]

We first fix a deterministic rule for resolving ties and define the \emph{most likely error decoder}
\[
 \widehat e(\sigma)\in\arg\max_{\partial e=\sigma}Q(e).
\]
This decoder does not necessarily maximize the probability of a full logical class. Optimal logical class decoding can therefore only perform better.

\begin{theorem}[Positive threshold by contour counting]
\label[theorem]{thm:peierls}
If $\kappa_0w(q)<1$, the most likely error decoder obeys
\begin{equation}
 p_{\rm fail}^{\sharp}\le
 n_k\frac{(\kappa_0w(q))^{d_{\min}}}{1-\kappa_0w(q)},
 \qquad
 p_{\rm fail}^{\ML}\le p_{\rm fail}^{\sharp}.
\label{eq:peierls_bound}
\end{equation}
It follows that every bounded geometry family with $d_{\min}/\log n_k\to\infty$ has a nonzero threshold. The same estimate applies to a repeated measurement spacetime complex after replacing $(n_k,d_{\min},\Delta_X)$ by the corresponding spacetime quantities; in particular, it gives a threshold for any lifted family satisfying $d_{\min}^{\rm st}/\log n_k^{\rm st}\to\infty$.
\end{theorem}

\begin{proof}
Let $e$ be the actual error and define $c=\widehat e(\partial e)-e$. A decoding failure implies that $c$ is a nontrivial homology cycle. Decompose its support into connected components of the incidence graph. Since cells that meet a common boundary cell are adjacent and $\partial c=0$, every labeled component is itself a cycle. At least one component, denoted $c_0$, must be homologically nontrivial. Otherwise their sum would be a boundary.

The errors $\widehat e$ and $\widehat e-c_0$ have the same syndrome. Since $\widehat e$ is a most likely error,
\[
 Q(\widehat e)\ge Q(\widehat e-c_0).
\]
All other connected components are disjoint from $c_0$. Product factorization then turns this inequality into $Q(e+c_0)\ge Q(e)$. For a fixed nonzero labeled cycle $c$,
\begin{align}
 \Prob[Q(e+c)\ge Q(e)]
 &\le\sum_e Q(e)\sqrt{\frac{Q(e+c)}{Q(e)}}\notag\\
 &=\sum_e\sqrt{Q(e)Q(e+c)}\notag\\
 &\le w(q)^{\abs{\supp c}}.
\label{eq:peierls_hellinger}
\end{align}
There are at most $(e\Delta_X)^\ell$ connected supports of size $\ell$ that contain a specified cell, and each support has at most $(\abs G-1)^\ell$ nonzero labelings. Summing over the first cell and over $\ell\ge d_{\min}$ gives \eqref{eq:peierls_bound}. The optimal logical decoder has failure probability no larger than this decoder.
\end{proof}

For the symmetric distribution on a group of order $m=\abs G$,
\[
 q(0)=1-p,
 \qquad
 q(g\ne0)=\frac{p}{m-1},
\]
one has
\begin{equation}
 w(p)=2\sqrt{\frac{p(1-p)}{m-1}}+\frac{m-2}{m-1}p.
\label{eq:flat_bhattacharyya}
\end{equation}

\subsection{Fractional moments and maximal noise}

For an actual error $e$, set
\[
 R_h(e)=\frac{Z_0(e+\gamma_h)}{Z_0(e)}.
\]
By \cref{thm:wilson_odds}, $R_h(e)$ is a positive logical odds ratio. In the dual model, it is represented by a Wilson or disorder amplitude.

\begin{proposition}[Decoder independent fractional moment bound]
\label[proposition]{prop:wilson_bound}
For every $s\in(0,1)$, the maximum likelihood failure probability satisfies
\begin{equation}
 p_{\rm fail}^{\ML}
 \le\sum_{h\ne0}\E_e\!\left[\min\{1,R_h(e)\}\right]
 \le\sum_{h\ne0}\E_e R_h(e)^s.
\label{eq:wilson_bound}
\end{equation}
On the Nishimori ensemble, however, the first moment is exactly
\begin{equation}
 \boxed{\E_e R_h(e)=1}
 \qquad(h\in H).
\label{eq:nishimori_first_moment}
\end{equation}
At $s=1/2$, the fractional moment is the Bhattacharyya affinity between the joint syndrome and logical distribution and its translate:
\begin{equation}
 \E_e\sqrt{R_h(e)}
 =\sum_{\sigma}\sum_{g\in H}
 \sqrt{Z_{\sigma,g}Z_{\sigma,g+h}}.
\label{eq:logical_bhattacharyya}
\end{equation}
It follows that if, for some fixed $s\in(0,1)$,
\begin{equation}
 \E R_h^s\le \e^{-\tau d_h}
 \quad(h\ne0),
 \qquad
 \log(\abs H-1)=o(\tau d_{\min}),
\label{eq:fractional_moment_threshold}
\end{equation}
then $p_{\rm fail}^{\ML}\to0$ exponentially. At uniform noise the posterior is exactly flat,
\begin{equation}
 p_{\rm fail}^{\ML}=1-\frac1{\abs H},
 \qquad
 I_c=0.
\label{eq:uniform_failure}
\end{equation}
\end{proposition}

\begin{proof}
If the maximum likelihood decoder chooses the wrong class, then $R_h(e)\ge1$ for some nonzero $h$. Applying the union bound together with $\mathbf1_{R\ge1}\le\min\{1,R\}\le R^s$ gives \eqref{eq:wilson_bound}.

Let $C_{\sigma,g}=e_\sigma+\gamma_g+B_k$. Its total probability is $Z_{\sigma,g}$. For any $e\in C_{\sigma,g}$,
\[
 R_h(e)=\frac{Z_{\sigma,g+h}}{Z_{\sigma,g}}.
\]
Therefore
\[
 \E R_h
 =\sum_{\sigma,g}Z_{\sigma,g}
 \frac{Z_{\sigma,g+h}}{Z_{\sigma,g}}
 =\sum_{\sigma,g}Z_{\sigma,g+h}=1.
\]
The same calculation with exponent one half gives \eqref{eq:logical_bhattacharyya}. The threshold statement then follows from \eqref{eq:wilson_bound}. For uniform noise, every logical coset in a fixed syndrome fiber has the same weight.
\end{proof}

\begin{remark}[The first moment]
Equation \eqref{eq:nishimori_first_moment} is a change of measure identity. In a correctable phase, rare errors with a large translated odds ratio compensate the typical errors with a small ratio. The annealed first moment therefore cannot diagnose decoding. Fractional moments, quantiles, and typical defect free energies can.
\end{remark}

The number of logical classes must also be controlled. A disjoint union of topological components can suppress defects on each component while the number of sectors grows too quickly for the total fugacity to vanish.

\subsection{Monotone degradation families}

\begin{definition}[Degradation family]
A one parameter family of Pauli channels $\cN_t$ is a degradation family if for $s\ge0$ there is a channel $\cM_{t,s}$ with
\[
 \cN_{t+s}=\cM_{t,s}\circ\cN_t.
\]
Convolution semigroups and the usual symmetric error families are examples.
\end{definition}

\begin{proposition}[Sharp operational threshold]
\label[proposition]{prop:degradation}
Every optimal recovery quantity that obeys data processing under channel composition is monotone along a degradation family. Examples include optimal entanglement fidelity and minimum recovery error in diamond norm. Thus
\[
 t_c=\sup\set{t:\text{the family is asymptotically decodable at }t}
\]
is a sharp operational threshold. The corresponding thermal theory need not have a continuous transition. The definition also applies to a first order transition or a crossover.
\end{proposition}

\begin{proof}
The channel $\cN_{t+s}$ is obtained from $\cN_t$ by adding more noise. Any recovery for the former can therefore be viewed as a recovery for the latter after an additional channel. Data processing prevents an improvement under this composition.
\end{proof}

This threshold is the topological defect transition of the decoding ensemble along the Nishimori family. Identifying it with a singularity of a quenched bulk free energy or with a thermal transition of $\Tcode$ requires additional dynamical input.

\subsection{Logical mixing subgroups}

We now take the logical group $H$ to be fixed and finite. A phase can mix a subgroup of logical fluxes while retaining the quotient information.

\begin{theorem}[Posterior condensation subgroup]
\label[theorem]{thm:posterior_subgroup}
Let $U\le H$. The posterior approaches the uniform distribution on a random coset $h_*(\sigma)+U$ if and only if
\begin{equation}
 \max_{u\in U}\Delta_\sigma(u)\longrightarrow0,
 \qquad
 \min_{g\notin U}\Delta_\sigma(g)\longrightarrow+\infty
\label{eq:partial_condensation}
\end{equation}
in probability, with the convention that the minimum over an empty set is $+\infty$. In that phase
\[
 p_{\rm succ}^{\ML}\longrightarrow\frac1{\abs U}.
\]
A character $\eta\in\widehat H$ survives exactly when it is trivial on $U$. The remaining dual logical algebra is the annihilator
\[
 U^\perp=\set{\eta\in\widehat H:\eta(u)=1\text{ for all }u\in U}.
\]
\end{theorem}

\begin{proof}
Equation \eqref{eq:partial_condensation} makes the weights in $h_*+U$ asymptotically equal and suppresses every weight outside that coset. Conversely, convergence to the uniform distribution on the coset gives these ratios. The statement about the dual algebra follows from character orthogonality on $U$.
\end{proof}

\subsection{Coefficient subgroup filtration}

For composite $G$, different quotients of the logical information can fail at different noise strengths. Let $K\le G$ and let $\pi_K:G\to G/K$ be the quotient map. We say that decoding succeeds at level $K$ when the decoder recovers the homology class of $\pi_K(e)$.

\begin{theorem}[Condensation filtration]
\label[theorem]{thm:coefficient_filtration}
Fix a code family and noise family.
\begin{enumerate}
\item If $K\le K'$, then level $K$ decodability implies level $K'$ decodability.
\item Decodability of the $G/K$ homological code under the pushforward noise implies level $K$ decodability.
\item Suppose both $H_k(X_L;\ZZ)$ and $H_{k-1}(X_L;\ZZ)$ are torsion free for every $L$, as on tori. Then the decodable coefficient subgroups are upward closed and closed under intersection. Hence there is a unique minimal subgroup $K_*(t)$, and the surviving logical algebra contains that of the $G/K_*(t)$ code.
\item Along a degradation family, $K_*(t)$ is nondecreasing. The memory can therefore fail through a chain of partial condensation transitions in the subgroup lattice of $G$.
\end{enumerate}
\end{theorem}

\begin{proof}
The first two statements follow from functoriality of the quotient. To treat intersections, use the coefficient map
\[
 G/(K\cap K')\longrightarrow G/K\oplus G/K',
\]
which is injective. Under the stated torsion free hypotheses, the universal coefficient theorem reduces to
\[
 H_k(X;A)\cong H_k(X;\ZZ)\otimes A.
\]
The injection of coefficient groups then induces an injection on $H_k$. Successful decoders at levels $K$ and $K'$ determine a unique common class in $H_k(X;G/(K\cap K'))$. Their failure probabilities can be combined with a union bound. The last statement follows by applying \cref{prop:degradation} to every quotient.
\end{proof}

The torsion assumption in \cref{thm:coefficient_filtration} is stronger than the assumption that only $H_k$ is torsion free. Torsion in $H_{k-1}$ produces the $\operatorname{Tor}$ term in the universal coefficient theorem and can create a Bockstein ambiguity invisible in the separate quotient sectors.

\begin{example}[$\mathbb Z_4$ partial failure]
Consider noise supported on $2\mathbb Z_4\cong\mathbb Z_2$. The quotient $\mathbb Z_4/2\mathbb Z_4$ is noiseless, while the remaining fine sector is a $\mathbb Z_2$ homological code. Once the $\mathbb Z_2$ sector passes its threshold, but before the noise acts on the quotient, the subgroup $2\mathbb Z_4$ is mixed and the surviving memory is the quotient $\mathbb Z_2$. Small odd shifts can then produce a second transition, giving three distinct phases.
\end{example}

\subsection{Matter and screened order parameters}

Dynamical matter can terminate flux lines. In the decoding model, this occurs when the noise action contains summed charged variables or correlated processes that do not preserve the CSS decomposition. A bare Wilson loop may then obey a perimeter law on both sides of the transition, as in the Fradkin--Shenker continuity region \cite{FradkinShenker}. The appropriate order parameter is the dressed defect free energy, equivalently a Fredenhagen--Marcu horseshoe ratio \cite{FredenhagenMarcu}.

The exact posterior remains well defined:
\[
 \frac{Z_{\sigma,h}}{Z_{\sigma,0}}
 =\e^{-\Delta F_{\sigma,h}}
\]
is the screened free energy cost of changing the logical completion. The information theoretic criterion can therefore remain nontrivial even when no local thermodynamic singularity or useful bare loop criterion exists. A sharp transition then becomes a separate dynamical question. If fundamental matter explicitly breaks the relevant higher form symmetry, the syndrome must include the matter endpoints or the code must be formulated in relative homology. Otherwise the topological distance can collapse.

\part{Topological memories and confining phases}

\section{\texorpdfstring{Four dimensional $\ZN$ toric code}{Four dimensional ZN toric code}}
\label{sec:3d_code}

\subsection{Hamiltonian and logical algebra}

Place a qudit with $N$ levels on every edge of a periodic cubic lattice of linear size $L$. Let $X$ and $Z$ be the clock and shift operators, with $ZX=\omega XZ$ and $\omega=\e^{2\pi\ii/N}$. The vertex and plaquette operators are
\[
 A_v=\prod_{\ell\ni v}X_\ell^{\pm1},
 \qquad
 B_p=\prod_{\ell\in\partial p}Z_\ell^{\pm1}.
\]
The commuting projector Hamiltonian is
\begin{equation}
 H_0=J_A\sum_v(1-\Pi_A(v))+J_B\sum_p(1-\Pi_B(p)),
\label{eq:3d_toric_hamiltonian}
\end{equation}
with
\[
 \Pi_A(v)=\frac1N\sum_{j=0}^{N-1}A_v^j,
 \qquad
 \Pi_B(p)=\frac1N\sum_{j=0}^{N-1}B_p^j.
\]
Violations of $A_v$ are point charges, while violated plaquettes form magnetic flux loops. On $T^3$, choose coordinate one cycles $\gamma_i$ and dual two tori $\Sigma_i$. The logical operators
\begin{equation}
 \overline Z_i=\prod_{\ell\in\gamma_i}Z_\ell,
 \qquad
 \overline X_i=\prod_{\ell\perp\Sigma_i}X_\ell,
 \qquad
 \overline X_i\overline Z_j=\omega^{\delta_{ij}}\overline Z_j\overline X_i
\label{eq:3d_logicals}
\end{equation}
give three Weyl pairs.

\begin{proposition}[The fixed point memory]
\label[proposition]{prop:fixed_code}
On the cubical three torus,
\begin{enumerate}
\item $H_0$ defines a $[[3L^3,3,L]]_N$ stabilizer code. The shortest line logical has weight $L$, while a dual membrane logical has weight $L^2$.
\item The model satisfies local topological order on length scales proportional to $L$ and has a uniform spectral gap bounded below by $\min(J_A,J_B)>0$ in the projector normalization.
\item Under sufficiently weak bounded finite range perturbations, the low energy band of rank $N^3$ remains separated by a nonzero gap. Its splitting is superpolynomially small in $L$, and its logical algebra is carried by quasiadiabatically dressed versions of \eqref{eq:3d_logicals} \cite{BHM}.
\item Under independent $\ZN$ Pauli noise, each CSS sector has a positive optimal threshold by \cref{thm:peierls}; measurement errors are included by the four dimensional spacetime lift. For the $N=2$ cubic model, the corresponding random one form and two form gauge theories and their optimal phenomenological thresholds were analyzed explicitly in \cite{Xu3DToric}.
\end{enumerate}
\end{proposition}

\begin{proof}
The homology group $H_1(T^3;\ZN)\cong(\ZN)^3$ gives three logical qudits. A nontrivial one cycle has length at least $L$, and a coordinate loop attains this bound. A nontrivial dual two cocycle has area at least $L^2$. The minimum Pauli weight is therefore $L$. An operator with trivial topology and no syndrome is a stabilizer, which gives local indistinguishability and the local consistency condition used in the stability theorem. Every excited state violates at least one projector, so the commuting projector gap is at least $\min(J_A,J_B)$. Global constraints can only make the first excitation more expensive. The perturbative statement follows from the Bravyi--Hastings--Michalakis theorem \cite{BHM}, and the threshold follows from \cref{thm:peierls}.
\end{proof}

The notation $[[3L^3,3,L]]_N$ records the number of edge qudits, the number of logical qudits, and the minimum distance. The distance is neither an energy barrier nor a particle mass.

\subsection{Distance, spectral gap, and energy barrier}

The two logical species have different energy barriers. A line logical can be implemented by creating two pointlike endpoints and moving one of them around a noncontractible cycle. Its maximum energy cost is $O(1)$. A membrane logical instead passes through configurations bounded by a flux loop. At the widest stage, this loop has length $\Theta(L)$, so the barrier grows linearly when the flux loop has positive line tension.

For the clock Hamiltonian
\[
 -J\sum_p(B_p+B_p^\dagger),
\]
the minimum nonzero plaquette violation cost is
\begin{equation}
 \epsilon_B=2J\left(1-\cos\frac{2\pi}{N}\right).
\label{eq:epsilon_B}
\end{equation}

\begin{corollary}[Wrapped flux energy]
\label[corollary]{cor:wrapped_flux}
Let $s$ be the closed $\ZN$ valued cochain of plaquette violations. If $[s]\ne0$, then every state in that sector satisfies
\begin{equation}
 E-E_0\ge\epsilon_B d^*,
 \qquad
 d^*=L
\label{eq:wrapped_flux_bound}
\end{equation}
on the cubic three torus.
\end{corollary}

\begin{proof}
Every nonzero plaquette violation costs at least $\epsilon_B$. The cube constraints imply that $s$ is closed. A nontrivial closed two cochain is Poincar\'e dual to a noncontractible dual line, whose length is at least $L$.
\end{proof}

Equation \eqref{eq:wrapped_flux_bound} bounds a wrapped flux sector. The ordinary Hamiltonian gap remains $O(\epsilon_B)$ because a local excitation has constant energy. Wrapped flux energy, code distance, and local particle gap are therefore distinct quantities.

\subsection{Memories at finite temperature}

The three dimensional toric code is not a fully self correcting quantum memory. Its point charge sector has an $O(1)$ barrier and cannot preserve quantum information passively at positive temperature. The flux loop sector can nevertheless retain a classical bit below the loop proliferation transition \cite{CastelnovoChamon}. More generally, passive stability depends on both dimension and excitation structure \cite{BravyiTerhal,Stahl}. We will therefore distinguish active decoding thresholds, passive memory times, and thermal transitions of the clean Hamiltonian.

The species rule in \cref{prop:replica_two} gives the same conclusion. At the fixed point, a closed logical fidelity loop is exactly one and does not diagnose the transition. An open probe in the conjugate species vanishes by superselection. The useful intrinsic observable is an open pair in the same species, normalized in the manner of Fredenhagen--Marcu.

\section{Composite and nonabelian sectors}
\label{sec:extensions}

\subsection{\texorpdfstring{Composite $N$ and global form filtrations}{Composite N and global form filtrations}}

For composite $N$, \cref{thm:coefficient_filtration} allows different divisor sectors to fail at different noise strengths. In the $\mathbb Z_4$ code, noise concentrated on doubled flux can destroy the fine $\mathbb Z_2$ sector while preserving the quotient sector. The intermediate memory is then $\mathbb Z_2$. Under the conditional Yang--Mills interpretation of \cref{sec:psu_code}, this resembles motion through the lattice of global forms between $\PSU(N)$ and $\SU(N)$. For example, a divisor biased $\PSU(4)$ channel can have an intermediate $\SU(4)/\mathbb Z_2$ logical algebra between two Nishimori transitions. This conclusion concerns the chosen channel and does not predict the phase structure of pure Yang--Mills theory.

\subsection{Finite error groupoids}
\label{sec:nonabelian}

The orbit argument can also be stated without an Abelian group law. Let a finite groupoid $\mathscr G\rightrightarrows\Omega$ describe physically trivial deformations among classically distinguishable errors. Its connected components are the logical classes.

\begin{proposition}[Finite error groupoid identity]
\label[proposition]{prop:groupoid}
For $x\in\Omega$, define
\begin{equation}
 Z_x=\frac{1}{\abs{\operatorname{Aut}(x)}}
 \sum_{\gamma:\,\operatorname{src}\gamma=x}
 P(\operatorname{tgt}\gamma).
\label{eq:groupoid_weight}
\end{equation}
Then
\begin{equation}
 Z_x=\sum_{y\in[x]}P(y).
\label{eq:groupoid_orbit}
\end{equation}
Maximum likelihood decoding again reduces to a comparison of orbit partition sums.
\end{proposition}

\begin{proof}
For every $y$ in the component of $x$, the set $\operatorname{Hom}(x,y)$ is a torsor for $\operatorname{Aut}(x)$. Thus each $y$ appears exactly $\abs{\operatorname{Aut}(x)}$ times in \eqref{eq:groupoid_weight}.
\end{proof}

As one restricted example, consider central flux noise in a quantum double $D(\Gamma)$. Take the errors to lie in the center $Z(\Gamma)$, and restrict to flat holonomies for which the action of $H^1(X;Z(\Gamma))$ is free. The syndrome and likelihood then reduce to those of a $Z(\Gamma)$ homological code. For $D(D_4)$ and $D(Q_8)$, the center is $\mathbb Z_2$, so this gives a protected central sector with the $\mathbb Z_2$ threshold. If the action is not free, stabilizer multiplicities make the logical dimension depend on the sector, and a tube algebra treatment is needed. Coherent nonabelian anyon noise also lies outside the scalar Gibbs description because syndrome measurement can preserve coherence in fusion spaces.

\subsection{Compact Abelian groups}

Replacing finite sums by Haar integrals gives a formal extension to compact Abelian groups and rotor codes \cite{RotorCodes}. The clean theories already suggest a sharp dimensional distinction. With faulty measurements, a two dimensional rotor memory gives a three dimensional compact $U(1)$ decoding theory. The Goepfert--Mack theorem shows that the clean theory confines at every coupling \cite{GoepfertMack}, which suggests a zero threshold for the lifted memory. We do not prove the necessary disordered and operator algebraic extension. In three spatial dimensions, the corresponding four dimensional compact $U(1)$ theory has a Coulomb phase \cite{Guth,FrohlichSpencer}, so a positive threshold is possible. The large $N$ behavior of $\mathbb Z_N$ codes also depends on the physical noise scale and cannot be inferred from the group order alone.

The positive scalar construction has clear limits. At nonzero theta angle, the Euclidean weight can be complex and cannot be used as an ordinary probability. Coherent nonabelian anyon errors also require operator valued sector weights because syndrome measurement may preserve coherence in fusion spaces. The groupoid identity applies to classical orbit problems, not to this operator valued setting.

\section{\texorpdfstring{Confining $\PSU(N)$ vacua and spatial codes}{Confining PSU(N) vacua and spatial codes}}
\label{sec:psu_code}

\subsection{Global form and generalized symmetry}

Consider pure $\SU(N)$ gauge theory on a spatial cubic lattice, with Kogut--Susskind Hamiltonian
\begin{equation}
 H_{\rm KS}=\frac{g^2}{2a}\sum_\ell E_\ell^2
 -\frac{1}{2g^2a}\sum_p\left(\tr U_p+\tr U_p^\dagger\right).
\label{eq:KS_parent}
\end{equation}
The simply connected theory has an electric $\ZN$ one form center symmetry. Gauging this symmetry changes the global form to $\PSU(N)=\SU(N)/\ZN$ and gives a dual magnetic $\ZN$ one form symmetry \cite{GKSW,KapustinThorngren,Tachikawa}. On a spatial three torus, the magnetic symmetry operators are surfaces $\eta(\Sigma_i)$, while the genuine charged 't Hooft operators are lines $T(\gamma_i)$. They obey
\begin{equation}
 \eta(\Sigma_i)T(\gamma_j)
 =\omega^{\delta_{ij}}T(\gamma_j)\eta(\Sigma_i).
\label{eq:psu_heisenberg}
\end{equation}
This is precisely the logical algebra of the $\ZN$ code in \eqref{eq:3d_logicals}.

The global form matters for the code interpretation. The $\SU(N)$ theory is expected to have a unique confining vacuum, so it does not itself give an $N^3$ dimensional topological ground space. The adjoint $\PSU(N)$ theory can support such a sector if its magnetic one form symmetry is spontaneously broken.

\subsection{Infrared code inheritance}

\begin{assumption}[Confining parent at fixed lattice spacing]
\label[assumption]{ass:ym_package}
At fixed lattice spacing and fixed bare coupling, assume that the parent $\SU(N)$ theory has the following properties uniformly in the size of the spatial three torus:
\begin{enumerate}
\item a unique gapped vacuum in the trivial flux sector;
\item unbroken electric center symmetry and a fundamental Wilson loop area law with string tension $\sigma>0$;
\item vanishing continuous and discrete theta angles and no obstructing anomaly, so gauging the center produces the untwisted finite gauge sector.
\end{enumerate}
\end{assumption}

These are the standard lattice properties associated with confinement through the Wilson loop criterion \cite{tHooftConfinement}. At sufficiently strong coupling, the local gauge invariant sector can be controlled as in \cref{sec:strong_psu}. The same uniform statements are not known at the continuum scaling point.

\begin{proposition}[Code in a confining $\PSU(N)$ vacuum]
\label[proposition]{prop:psu_inheritance}
Assume \cref{ass:ym_package}, and assume that gauging an unbroken finite one form symmetry of a unique gapped phase produces the corresponding untwisted finite gauge theory in the infrared. Then the confining $\PSU(N)$ vacuum on $T^3$ has the following properties.
\begin{enumerate}
\item a topological ground state sector of dimension $N^3$, up to finite size splitting;
\item logical algebra generated by the wrapped 't Hooft lines and magnetic symmetry surfaces in \eqref{eq:psu_heisenberg};
\item local indistinguishability on contractible regions of diameter proportional to $L$, together with a quasilocally dressed logical algebra and an approximate correctability length proportional to $L$;
\item an $O(1)$ energy barrier for the line logical and an $\Omega(L)$ barrier for the conjugate membrane logical when the magnetic flux loop has positive line tension.
\end{enumerate}
Thus its magnetic infrared sector is in the phase of the $[[3L^3,3,L]]_N$ code of \cref{prop:fixed_code}.
\end{proposition}

\begin{proof}[Argument]
After gauging, unbroken electric center symmetry in a unique gapped phase becomes broken dual magnetic symmetry. The long distance topological theory is the untwisted $\ZN$ finite gauge theory. Its two form flux basis is labeled by $H^2(T^3;\ZN)$. The dual basis associated with the broken magnetic symmetry is Poincar\'e dual to $H^1(T^3;\ZN)$. Both spaces have $N^3$ elements, and their line and surface operators obey \eqref{eq:psu_heisenberg}. Stability of a gapped topological phase dresses the fixed point logical operators by quasilocal unitaries and preserves macroscopic local indistinguishability. It does not, however, assign a literal Pauli weight to an operator in the microscopic Yang--Mills Hilbert space. In the magnetic description, the line logical can be implemented with pointlike endpoints, while the conjugate membrane sweeps a flux loop of length $\Theta(L)$.
\end{proof}

The fixed point code in \cref{prop:fixed_code} is independent of this Yang--Mills interpretation. The inheritance statement in \cref{prop:psu_inheritance} uses \cref{ass:ym_package} and the infrared symmetry argument. None of these assumptions enter the Euclidean center sheet construction in \cref{sec:finite_k_model}.

\subsection{Dimensional dependence}

\begin{lemma}[Dimensional obstruction for the adjoint vacuum mechanism]
\label[lemma]{lem:dimensional_obstruction}
In $d$ spacetime dimensions, pure $\PSU(N)$ gauge theory has a magnetic $\ZN$ symmetry of form degree $d-3$. If the confining vacuum breaks this symmetry, the ground state sectors on $M_{d-1}$ are labeled by $H^{d-3}(M_{d-1};\ZN)$. For $d=3$, the order parameter is local and the degeneracy gives a classical $\ZN$ register. For $d\ge4$, the order parameter and its conjugate are extended and obey a nontrivial Heisenberg algebra, so the sectors form a quantum code.
\end{lemma}

\begin{proof}
Magnetic defects have codimension three, so their worldvolume has dimension $d-3$. When a zero form symmetry breaks, its vacua are locally distinguishable and there is no protected conjugate basis. A broken higher form symmetry instead has extended charged operators and symmetry operators. Their intersection pairing gives the noncommuting logical algebra.
\end{proof}

This observation concerns the magnetic symmetry of pure adjoint Yang--Mills theory. It does not rule out other lower dimensional topological codes.

\subsection{\texorpdfstring{Uniform magnetic phases and the sign for even $N$}{Uniform magnetic phases and the sign for even N}}

The strong coupling center Hamiltonian can generate a magnetic term with a uniform phase,
\[
 -J\sum_\alpha\left(\omega^sM_\alpha+\omega^{-s}M_\alpha^\dagger\right).
\]
A depth one product of $X$ operators changes this phase by a coboundary. We prove the exact statement, including the finite volume obstruction, in \cref{app:uniform_phase}. In the two form center description, the phase can be removed when $N\mid sL^3$. A real strong coupling coefficient has sign $\pm1$. On these volumes, its sign is therefore a gauge choice when $N$ is even. This does not show that the coefficient is nonzero, and the remaining terms must still be controlled. For odd $N$, including $N=3$, the sign cannot be removed by this argument.

\section{Strong coupling reduction of the spatial Hamiltonian}
\label{sec:strong_psu}

\subsection{Local strong coupling sector}

At small Euclidean Wilson coupling, the character expansion is a convergent gas of closed plaquette polymers. Standard strong coupling methods then give analyticity, a Wilson loop area law, reflection positivity, and exponential clustering for local gauge invariant operators \cite{OsterwalderSeiler,Munster}. A flat center two form background changes the polymer activities only by phases. The absolute convergence estimates are therefore uniform over the topological flux sectors.

\begin{proposition}[Strong coupling local core]
\label[proposition]{prop:strong_local_core}
For every $N$, there is a $\beta_0(N)>0$ such that the following statements hold for $0<\beta<\beta_0(N)$.
\begin{enumerate}
\item the free energy density is analytic and the cluster expansion is uniform in volume and in flat center flux sector;
\item connected local gauge invariant correlators decay exponentially;
\item fundamental Wilson loops obey an area law;
\item reflection positivity gives a positive transfer matrix, and the local operator channel has a positive glueball scale mass.
\end{enumerate}
These conclusions concern the center neutral local sector. They do not by themselves establish a topological gap after the center is gauged.
\end{proposition}

\begin{proof}[Sketch]
Expand each plaquette Boltzmann factor in irreducible characters. Link integration leaves closed spin foam polymers, whose activities are bounded by a constant times $\beta$ raised to their area. The Kotecky--Preiss criterion is satisfied for sufficiently small $\beta$. Since a flat center background contributes only center phases, it does not change the absolute activity bounds. Reflection positivity gives a positive transfer matrix, and exponential clustering gives a positive mass in every local operator channel with nonzero overlap.
\end{proof}

In the Hamiltonian formulation, finite volume Kato perturbation theory about the electric vacuum gives a unique ground state and a gap of order $g^2/a$. The lightest gauge invariant excitation is the shorter of an elementary four link electric loop and a wrapped fundamental line of length $L$. The wrapped line is lighter when $L\le3$. The fundamental string tension is of order $g^2/a^2$. To make these statements uniform in volume for the constrained infinite dimensional nonabelian Hamiltonian, one would still need a Hamiltonian cluster expansion.

\subsection{Effective center Hamiltonian}

After the center is gauged, the electric strong coupling limit has $N^3$ flux sector vacua on $T^3$. Let $P$ project onto this band, and let $V$ denote the magnetic plaquette perturbation. Schrieffer--Wolff perturbation theory gives
\begin{equation}
 H_{\rm eff}=PVP+PV(E_0-H_E)^{-1}VP+\cdots.
\label{eq:SW_effective}
\end{equation}
The first term vanishes because one plaquette insertion creates electric flux. A local term that acts within the flux band must be invariant under the center one form gauge symmetry. In the two form description, the smallest such product is the six plaquette boundary of an elementary cube. The leading local term therefore has the form
\begin{equation}
 -J_B\sum_c\left(B_c+B_c^\dagger\right),
 \qquad
 B_c=\prod_{p\subset\partial c}b_p,
\label{eq:center_cube_term}
\end{equation}
and cannot appear before sixth order on a periodic lattice with $L\ge3$. When $L=2$, a wrapped coordinate plane contains only four distinct plaquettes and produces an earlier finite size splitting. For this reason, the smallest clean numerical test has $L\ge3$.

The center factor for the elementary cube process is the positive group integral
\begin{equation}
 \int\prod_{\ell\subset c}\dd U_\ell
 \prod_{p\subset\partial c}\chi_f(U_p)=\frac1{N^4}.
\label{eq:cube_integral}
\end{equation}
Each of the twelve links carries one fundamental and one antifundamental matrix element. Haar integration gives a factor $1/N$ for each edge, and the remaining contraction has eight index loops.

\subsection{Uniformity conditions}

The strong coupling construction would establish the spatial code at fixed lattice spacing once the following two estimates are proved.

\begin{description}[leftmargin=3.7em,style=nextline]
\item[(R1)] The Schrieffer--Wolff series converges uniformly in $L$, the cube coefficient $J_B$ is nonzero, and all remaining terms stay inside the topological stability radius. For even $N$ and divisible volumes, \cref{app:uniform_phase} removes the sign of $J_B$ as an independent issue.
\item[(R2)] The effective center Hamiltonian lies in the untwisted deconfined $\ZN$ phase rather than a confined or twisted phase, equivalently its genuine 't Hooft line has a perimeter law and its lightest topological excitation has a positive gap of order $\abs{J_B}$.
\end{description}

The main difficulties are the unbounded electric Casimir, the Gauss law constraint, and the nonabelian recoupling coefficients for virtual closed magnetic processes. The finite $\kappa$ Euclidean model in the next part avoids these issues. It uses the full Yang--Mills path integral as a hidden variable likelihood and does not identify the microscopic vacuum with the code.

\part{Finite curvature Yang--Mills decoding}

\section{\texorpdfstring{Finite curvature $\SU(N)$ center sheet model}{Finite curvature SU(N) center sheet model}}
\label{sec:finite_k_model}

The spatial construction in the previous part asks whether a microscopic $\PSU(N)$ Hamiltonian flows to a topological code. We now take a different approach. We define the code algebraically and use the Euclidean $\SU(N)$ path integral as a positive hidden variable model for correlated Pauli errors. The resulting identity is exact at finite volume for every Wilson coupling and every finite curvature fugacity. No infrared reduction to a finite gauge theory is required.

\subsection{Two form code}

Let $\Lambda$ be a finite oriented four dimensional cell complex. We will usually take it to be a periodic hypercubic lattice. Consider the cochain complex
\begin{equation}
 C^1(\Lambda;\ZN)\xrightarrow{\dd_1}
 C^2(\Lambda;\ZN)\xrightarrow{\dd_2}
 C^3(\Lambda;\ZN),
 \qquad \dd_2\dd_1=0.
\label{eq:euclidean_cochain_complex}
\end{equation}
We place a qudit with $N$ levels on every plaquette. A $Z$ type sheet error is labeled by
\[
 B\in C^2(\Lambda;\ZN),
 \qquad Z(B)=\prod_p Z_p^{B_p}.
\]
Its syndrome and stabilizer equivalence are
\begin{equation}
 J=\dd B\in C^3(\Lambda;\ZN),
 \qquad
 B\sim B+\dd\lambda,
 \quad \lambda\in C^1(\Lambda;\ZN).
\label{eq:euclidean_code_data}
\end{equation}
For a fixed syndrome $J$, the logical completions form a torsor for
\begin{equation}
 \cL_\Lambda=H^2(\Lambda;\ZN)=\ker\dd_2/\im\dd_1.
\label{eq:euclidean_logical_group}
\end{equation}
On $T^4$, $\cL_\Lambda\cong(\ZN)^6$. The same construction works on a slab by using relative cohomology. This allows physical boundaries and prescribed defect worldlines.

This Euclidean code defines an inference problem for center vortex sheets. The sheet boundary is the measured syndrome. Under Poincar\'e duality, $\star J$ is a conserved one dimensional current on the dual lattice. The purpose of the construction is to formulate this inference problem, not to propose a four dimensional hardware architecture.

\subsection{Center twisted Wilson action}

Place an $\SU(N)$ matrix $U_\ell$ on each oriented link, with $U_{-\ell}=U_\ell^{-1}$, and let $U_p$ denote the oriented plaquette holonomy. We write
\[
 \omega=\exp(2\pi\ii/N).
\]
For any $B\in C^2(\Lambda;\ZN)$, whether flat or not, define
\begin{equation}
 \cZ_\beta[B]
 =\int\prod_\ell\dd U_\ell\,
 \exp\left[
  \frac{\beta}{N}\sum_p
  \operatorname{Re}\tr\!\left(\omega^{-B_p}U_p\right)
 \right].
\label{eq:ym_partition_B_merged}
\end{equation}
When $B$ is flat, this is the Wilson partition function in a background two form gauge field for the electric $\ZN$ one form symmetry. Equivalently, it is a fixed 't Hooft flux sector \cite{tHooftFlux,GKSW,AbeEtAl}. When $B$ is not flat, the same positive local integral remains well defined and inserts a center sheet whose boundary is $\star\dd B$.

\begin{lemma}[One form orbit invariance]
\label[lemma]{lem:one_form_orbit}
For every $\lambda\in C^1(\Lambda;\ZN)$,
\begin{equation}
 \cZ_\beta[B+\dd\lambda]=\cZ_\beta[B].
\label{eq:one_form_orbit_invariance}
\end{equation}
\end{lemma}

\begin{proof}
Use the normalized Haar change of variables
\[
 U_\ell\longmapsto\omega^{\lambda_\ell}U_\ell.
\]
Then $U_p\mapsto\omega^{(\dd\lambda)_p}U_p$. The combination $\omega^{-B_p}U_p$ is therefore unchanged when $B\mapsto B+\dd\lambda$.
\end{proof}

Thus stabilizer equivalence in the code is the same as one form gauge equivalence in the Yang--Mills integral.

\subsection{Curvature fugacity and local parent model}

Choose a strictly positive local function $r_\kappa:\ZN\to\RR_{>0}$, and define
\begin{equation}
 \nu_\kappa(J)=\prod_{c\in\Lambda_3}r_\kappa(J_c).
\label{eq:curvature_fugacity}
\end{equation}
We will often use the center symmetric choice
\begin{equation}
 r_\kappa(j)=\exp\left[\kappa\cos\left(\frac{2\pi j}{N}\right)\right].
\label{eq:cosine_fugacity}
\end{equation}
Large positive $\kappa$ suppresses sheet boundaries. At any finite $\kappa$, every realizable syndrome still has nonzero probability.

We define the correlated Pauli distribution and the corresponding channel by
\begin{align}
 P_{\beta,\kappa}(B)
 &=\frac{1}{\cN_{\beta,\kappa}}
 \nu_\kappa(\dd B)\cZ_\beta[B],
\label{eq:center_sheet_probability}\\
 \cE_{\beta,\kappa}(\rho)
 &=\sum_{B\in C^2}P_{\beta,\kappa}(B)
 Z(B)\rho Z(B)^\dagger.
\label{eq:center_sheet_channel}
\end{align}
The marginal distribution of $B$ can be long range correlated. It is nevertheless obtained from the positive finite range field model
\begin{equation}
 \dd\Prob_{\beta,\kappa}(U,B)
 =\frac{1}{\cN_{\beta,\kappa}}
 \left(\prod_\ell\dd U_\ell\right)
 \exp\left[
 \frac{\beta}{N}\sum_p\operatorname{Re}\tr(\omega^{-B_p}U_p)
 \right]\nu_\kappa(\dd B).
\label{eq:joint_local_parent}
\end{equation}
The nonabelian link field therefore remains part of the local parent model. We have not replaced it by an effective $\ZN$ action.

For $N=2$, writing $\sigma_p=(-1)^{B_p}$ gives
\begin{equation}
 \exp\left[
 \frac{\beta}{2}\sum_p\sigma_p\tr U_p
 +\kappa\sum_c\prod_{p\subset\partial c}\sigma_p
 \right],
\label{eq:su2_villain_merged}
\end{equation}
up to an overall normalization. The cube product is the $\Ztwo$ monopole variable. Related auxiliary plaquette fields and monopole suppression terms appear in Villain formulations of $\SO(3)$ and $\SU(2)$ lattice gauge theory \cite{HallidayPhase,HallidayMonopoles,deForcrandJahn}. In the present construction, the plaquette field is an error label, the cube variable is the syndrome, and the global completion is the logical value.

\begin{figure}[t]
\centering
\begin{tikzpicture}[node distance=1.28cm and 1.55cm,>=Latex,thick]
\tikzset{
 qnode/.style={draw=deepblue,rounded corners,fill=softblue,minimum width=3.18cm,minimum height=0.95cm,align=center},
 hnode/.style={draw=darkgray,rounded corners,fill=softgray,minimum width=3.18cm,minimum height=0.95cm,align=center},
 lab/.style={font=\small,align=center}
}
\node[qnode] (Bq) {$B\in C^2(\Lambda;\ZN)$\\plaquette Pauli sheet};
\node[qnode,below=of Bq] (Jq) {$J=\dd B\in C^3$\\measured syndrome};
\node[qnode,below=of Jq] (Hq) {$[B]\in H^2(\Lambda;\ZN)$\\logical completion};
\node[hnode,right=3.35cm of Bq] (Bh) {center twisted plaquettes\\two form background};
\node[hnode,below=of Bh] (Jh) {dual conserved current $\star J$\\monopole worldlines};
\node[hnode,below=of Jh] (Hh) {'t Hooft flux\\topological sector};
\draw[->] (Bq)--node[left,lab]{$\dd$}(Jq);
\draw[->] (Jq)--node[left,lab]{choose a\\global completion}(Hq);
\draw[->] (Bh)--node[right,lab]{$\dd$}(Jh);
\draw[->] (Jh)--node[right,lab]{choose a\\topological sheet}(Hh);
\draw[<->,deepblue] (Bq)--node[above,lab]{same variable}(Bh);
\draw[<->,deepblue] (Jq)--node[above,lab]{Poincar\'e duality}(Jh);
\draw[<->,deepblue] (Hq)--node[above,lab]{posterior $\propto\cZ_\beta$}(Hh);
\end{tikzpicture}
\caption{The finite $\kappa$ decoding dictionary. The nonabelian links are hidden variables in the correlated Pauli distribution and are not code qudits.}
\label{fig:finite_k_dictionary}
\end{figure}
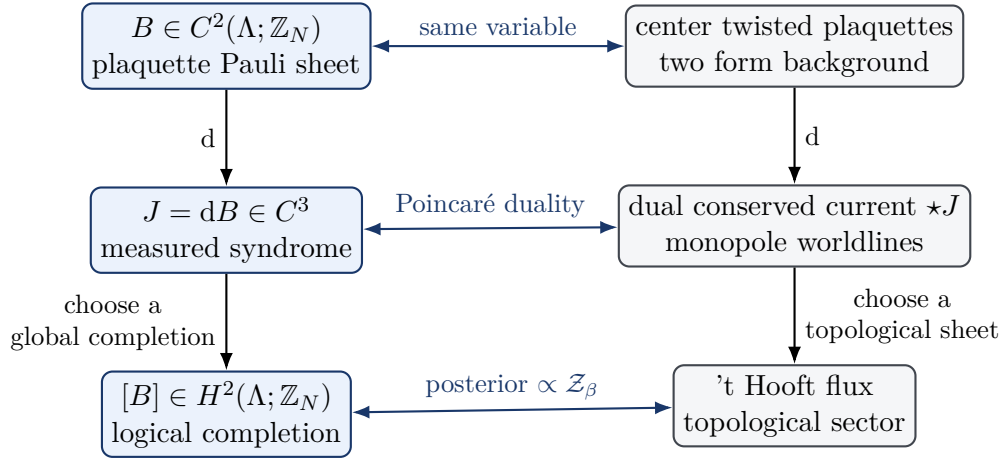

\subsection{Logical posterior and failure probability}

\begin{theorem}[Finite $\kappa$ Yang--Mills decoding identity]
\label[theorem]{thm:finite_k_posterior}
Let $J\in\im\dd_2$. Choose $B_J$ with $\dd B_J=J$, and choose cocycle representatives $B_h$ for $h\in\cL_\Lambda$. Then the channel in \eqref{eq:center_sheet_channel} obeys
\begin{equation}
 \boxed{
 \Prob_{\beta,\kappa}(h\mid J)=
 \frac{\cZ_\beta[B_J+B_h]}
 {\displaystyle\sum_{k\in\cL_\Lambda}\cZ_\beta[B_J+B_k]}.}
\label{eq:finite_k_posterior}
\end{equation}
Conditioning on the measured current $J$ removes the dependence on $\kappa$.

If $h_*(J)$ maximizes the numerator and
\begin{equation}
 \Delta F_J(g)=
 -\log\frac{\cZ_\beta[B_J+B_{h_*+g}]}
 {\cZ_\beta[B_J+B_{h_*}]},
\label{eq:ym_defect_free_energy}
\end{equation}
then
\begin{equation}
 \boxed{
 p_{\rm fail}^{\ML}(J)=
 \frac{\displaystyle\sum_{g\ne0}\e^{-\Delta F_J(g)}}
 {\displaystyle1+\sum_{g\ne0}\e^{-\Delta F_J(g)}}.}
\label{eq:finite_k_failure}
\end{equation}
\end{theorem}

\begin{proof}
The solutions of $\dd B=J$ form the affine space $B_J+\ker\dd_2$. Decompose $\ker\dd_2$ into cohomology classes modulo $\im\dd_1$. The unnormalized weight of the class $h$ is
\[
 W_{\beta,\kappa}(J,h)=
 \sum_{b\in\im\dd_1}P_{\beta,\kappa}(B_J+B_h+b).
\]
For exact $b$, we have $\dd b=0$, and \cref{lem:one_form_orbit} gives
\[
 \nu_\kappa(\dd(B_J+B_h+b))=\nu_\kappa(J),
 \qquad
 \cZ_\beta[B_J+B_h+b]=\cZ_\beta[B_J+B_h].
\]
All terms in the orbit sum are equal. Hence
\begin{equation}
 W_{\beta,\kappa}(J,h)=
 \frac{\abs{\im\dd_1}}{\cN_{\beta,\kappa}}
 \nu_\kappa(J)\cZ_\beta[B_J+B_h].
\label{eq:joint_sector_finite_k}
\end{equation}
Normalizing over $h$ proves \eqref{eq:finite_k_posterior}. Equation \eqref{eq:finite_k_failure} then follows from \cref{thm:exact_failure}.
\end{proof}

\begin{remark}[Two roles of $\kappa$]
At fixed syndrome, $\kappa$ does not change the relative topological completions. It changes the averaged failure probability
\[
 p_{\rm fail}^{\ML}(\beta,\kappa)
 =\sum_J\Prob_{\beta,\kappa}(J)p_{\rm fail}^{\ML}(J)
\]
because it changes the distribution of defect currents. The finite $\kappa$ phase diagram is therefore not determined by the thermodynamics of flat backgrounds alone.
\end{remark}

At finite $\kappa$, decoding compares the Yang--Mills weights of global center vortex sheets that end on the observed defect worldlines.

\section{Strong coupling dynamics of the center sheet model}
\label{sec:finite_k_strong_dynamics}

The posterior identity leaves the phase structure dynamical. At sufficiently small $\beta$, the convergent character expansion gives additional dynamical information. Integrating out the nonabelian links produces a positive local penalty for center sheet curvature, beginning at sixth order in the fundamental character coefficient. The same expansion gives exponential decay of local syndrome correlations, while the conditional distribution of global logical completions approaches the uniform distribution. Thus local syndrome correlations and global logical information can have different infrared behavior.

\subsection{Character expansion and effective syndrome action}

Write the normalized character expansion for one Wilson plaquette as
\begin{equation}
 \exp\!\left[\frac{\beta}{N}\operatorname{Re}\tr U\right]
 =c_0(\beta)\left[1+\sum_{\rho\ne\mathbf 1}
 d_\rho a_\rho(\beta)\chi_\rho(U)\right],
 \qquad u(\beta):=a_f(\beta),
\label{eq:focused_character_expansion}
\end{equation}
where $f$ denotes the fundamental representation. If $q_\rho\in\ZN$ is the $N$ ality of $\rho$, then
\begin{equation}
 \chi_\rho(\omega^{-B_p}U_p)
 =\omega^{-q_\rho B_p}\chi_\rho(U_p).
\label{eq:center_phase_character}
\end{equation}
For $\SU(2)$, $u=I_2(\beta)/I_1(\beta)=\beta/4+O(\beta^3)$; for $N\ge3$, $u=\beta/(2N^2)+O(\beta^2)$.

Link integration imposes conservation of $N$ ality. Every connected spin foam polymer $\gamma$ therefore carries a closed $\ZN$ valued plaquette cycle $q_\gamma$. Its dependence on the background is the phase
\begin{equation}
 \omega^{-\langle B,q_\gamma\rangle}.
\label{eq:polymer_background_phase}
\end{equation}
If $q_\gamma=\partial V_\gamma$ is contractible, then
\begin{equation}
 \langle B,q_\gamma\rangle
 =\langle \dd B,V_\gamma\rangle
 =\langle J,V_\gamma\rangle.
\label{eq:stokes_syndrome_polymer}
\end{equation}
Every contractible polymer is therefore a local function of the measured syndrome $J$. Only a homologically nontrivial polymer can distinguish two logical completions at fixed $J$.

Let $A_{\rm sys}(\Lambda)$ be the minimum plaquette area of a nonzero class in $H_2(\Lambda;\ZN)$. On an isotropic four torus of linear size $L$, we have $A_{\rm sys}=L^2$.

\begin{theorem}[Syndrome action and topological remainder at strong coupling]
\label{thm:strong_coupling_syndrome_action}
For every $N$, there are positive constants $u_0,C,\tau,$ and $\mu$ such that the following statements hold when $|u(\beta)|<u_0$ and every periodic extent is at least three.
\begin{equation}
 \log \cZ_\beta[B]
 =C_\beta(\Lambda)+\cF_\beta[J]+\cR_\beta[B],
 \qquad J=\dd B.
\label{eq:local_plus_topological_decomposition}
\end{equation}
Here $\cF_\beta$ is an exponentially local, translation covariant potential of $J$, and for every flat $h\in Z^2(\Lambda;\ZN)$,
\begin{equation}
 \sup_B\abs{\cR_\beta[B+h]-\cR_\beta[B]}
 \le C\abs{\Lambda}\,\e^{-\tau A_{\rm sys}(\Lambda)}.
\label{eq:topological_remainder_bound}
\end{equation}
The leading local term that depends on the background is the elementary cube:
\begin{equation}
 \cF_\beta[J]
 =\lambda_N(\beta)\sum_c
 \cos\!\left(\frac{2\pi J_c}{N}\right)
 +\sum_X\Phi_{\beta,X}(J|_X),
\label{eq:syndrome_effective_action}
\end{equation}
with
\begin{equation}
 \lambda_N(\beta)=
 \begin{cases}
 4u(\beta)^6+O(u^8),&N=2,\\[1mm]
 2N^2u(\beta)^6+O(u^7),&N\ge3,
 \end{cases}
\label{eq:induced_curvature_coupling}
\end{equation}
and, after the displayed cube term is removed,
\begin{equation}
 \sup_c\sum_{X\ni c}\e^{\mu\operatorname{diam}X}
 \norm{\Phi_{\beta,X}}_\infty=O(u^7)
\label{eq:local_potential_norm}
\end{equation}
(with $O(u^8)$ in the $\SU(2)$ case).
\end{theorem}

\begin{proof}
Insert \eqref{eq:focused_character_expansion} on every plaquette. Haar integration at a link projects the incident representations onto their invariant subspace, so the $N$ ality labels form closed plaquette cycles. The strong coupling cluster expansion converges absolutely for $|u|<u_0$ \cite{OsterwalderSeiler,Munster,LangelageEtAl}. The background phases in \eqref{eq:polymer_background_phase} have unit modulus and do not change the convergence constants. By \eqref{eq:stokes_syndrome_polymer}, every contractible cluster contributes to an exponentially local potential of $J$. A cluster whose phase changes under $B\mapsto B+h$ must contain a nontrivial $N$ ality cycle and therefore has at least $A_{\rm sys}$ plaquettes. The rooted cluster bound then gives \eqref{eq:topological_remainder_bound}.

The smallest nonzero contractible plaquette cycle on a hypercubic lattice is the oriented boundary of an elementary cube. Put the fundamental representation on its six faces. Each of the twelve links then carries one fundamental and one antifundamental matrix element, so
\begin{equation}
 \int\prod_{\ell\subset c}\dd U_\ell
 \prod_{p\subset\partial c}\chi_f(U_p)
 =N^{-4}.
\label{eq:strong_cube_integral}
\end{equation}
The six character coefficients give $(Nu)^6$. For $N\ge3$, the two orientations have complex conjugate phases and combine to give $2N^2u^6\cos(2\pi J_c/N)$. For $N=2$, the fundamental representation is self conjugate, and the single unoriented surface gives $4u^6(-1)^{J_c}$. Every other connected polymer with background dependence has higher activity. This proves \eqref{eq:induced_curvature_coupling} through \eqref{eq:local_potential_norm}.
\end{proof}

Combining \eqref{eq:local_plus_topological_decomposition} with the bare curvature fugacity, we find the marginal sheet measure
\begin{equation}
 \Prob_{\beta,\kappa}(B)
 \propto
 \exp\!\left[
 (\kappa+\lambda_N(\beta))\sum_c
 \cos\!\left(\frac{2\pi J_c}{N}\right)
 +\sum_X\Phi_{\beta,X}(J|_X)+\cR_\beta[B]
 \right].
\label{eq:strong_coupling_marginal_B}
\end{equation}
The coefficient in \eqref{eq:induced_curvature_coupling} is positive. Thus strong coupling Yang--Mills fluctuations suppress center monopole syndrome defects.

\subsection{\texorpdfstring{The $\beta=0$ phase transition}{The beta=0 phase transition}}

At $\beta=0$, $\cZ_0[B]=1$, and the $\SU(N)$ links decouple. For $N=2$, the remaining model is the four dimensional two form Ising gauge theory
\begin{equation}
 Z_{\text{2 form}}(\kappa)
 =\sum_{b_p=\pm1}
 \exp\!\left[\kappa\sum_c\prod_{p\subset\partial c}b_p\right].
\label{eq:two_form_ising_gauge}
\end{equation}
Wegner duality maps it to the nearest neighbor four dimensional Ising model with coupling $K$ satisfying \cite{Wegner}
\begin{equation}
 \e^{-2K}=\tanh\kappa.
\label{eq:twoform_ising_duality}
\end{equation}
Using the large volume estimate for $K_c$ from \cite{LundowMarkstrom2023} gives
\begin{equation}
 \kappa_c(0)
 =-\frac12\log\tanh K_c
 =0.953297052(33).
\label{eq:kappa_critical_anchor}
\end{equation}
The duality relation gives an exact anchor for the phase diagram of the finite $\kappa$ parent theory, with the quoted number inherited from the numerical Ising critical coupling. This transition is not a decoding threshold. Since $\cZ_0[B]$ is independent of $B$, the logical posterior is uniform on both sides of \eqref{eq:kappa_critical_anchor}.

Turning on $\beta$ shifts the leading local coupling by $\kappa\mapsto\kappa+4u^6$. If the critical manifold through $(0,\kappa_c(0))$ is differentiable under this symmetry preserving perturbation, its strong coupling tangent is
\begin{equation}
 \kappa_c(\beta)
 =\kappa_c(0)-4u(\beta)^6+O(u^8)
 =0.953297052(33)-\frac{\beta^6}{1024}+O(\beta^8).
\label{eq:formal_critical_tangent}
\end{equation}
Equation \eqref{eq:formal_critical_tangent} is a local strong coupling result. Determining the global critical curve and its universality class requires a continuation beyond the series. Standard expansions about both phases show that they persist in open regions at sufficiently small $\beta$, so a phase boundary remains between them.

\subsection{Logical sector mixing}

\begin{theorem}[Strong coupling logical mixing]
\label{thm:strong_coupling_logical_mixing}
Let $K_L=\abs{H^2(\Lambda_L;\ZN)}$ and set
\begin{equation}
 \varepsilon_L=C\abs{\Lambda_L}\e^{-\tau A_{\rm sys}(\Lambda_L)}.
\label{eq:strong_mixing_epsilon}
\end{equation}
In the domain of \cref{thm:strong_coupling_syndrome_action}, for every finite $\kappa$, every realizable syndrome $J$, and every logical completion $h$,
\begin{equation}
 \frac{\e^{-2\varepsilon_L}}{K_L}
 \le \Prob(h\mid J)
 \le \frac{\e^{2\varepsilon_L}}{K_L}.
\label{eq:posterior_uniform_bound}
\end{equation}
Hence, whenever $\varepsilon_L\to0$,
\begin{equation}
 p_{\rm succ}^{\ML}(J)=\frac1{K_L}+O(\varepsilon_L)
\label{eq:strong_coupling_success}
\end{equation}
uniformly in $J$ and $\kappa$.  On a fixed topology four torus this tends to $1/N^6$.
\end{theorem}

\begin{proof}
At fixed $J$, the term $C_\beta+\cF_\beta[J]$ in \eqref{eq:local_plus_topological_decomposition} is the same for every completion. Equation \eqref{eq:topological_remainder_bound} therefore bounds the difference between the logarithms of any two sector weights by $2\varepsilon_L$. Normalizing $K_L$ positive weights with pairwise ratios between $\e^{-2\varepsilon_L}$ and $\e^{2\varepsilon_L}$ gives \eqref{eq:posterior_uniform_bound}. Taking the largest weight gives \eqref{eq:strong_coupling_success}.
\end{proof}

This theorem adds dynamical information to the exact posterior identity. In the strong coupling confining region, contractible spin foams determine the local syndrome distribution. Only a wrapping surface can detect the logical center flux, and its activity is exponentially small in the systolic area. The syndrome sector may undergo the transition anchored at \eqref{eq:kappa_critical_anchor}, while the conditional logical sector remains asymptotically uniform throughout the convergent region at small $\beta$.

\subsection{Syndrome correlations}

Let $J_{k,r}$ be the conserved pair current of \cref{sec:syndrome_monopoles}, with two parallel worldlines separated by $r$. Let $\cC_k(r)$ be the core renormalized likelihood. The same local expansion controls this disorder correlator.

\begin{theorem}[Strong coupling syndrome decoupling]
\label{thm:strong_coupling_syndrome_decoupling}
For $|u(\beta)|<u_0$, there are constants $A,m_{\rm pol}>0$ such that, after taking the transverse thermodynamic limit in a fixed pure phase,
\begin{equation}
 \abs{\log\cC_k(r)-\log\cC_k(\infty)}
 \le A\e^{-m_{\rm pol}r}.
\label{eq:polymer_syndrome_decoupling}
\end{equation}
At finite periodic volume there is an additional correction bounded by
$C\abs{\Lambda}\e^{-\tau A_{\rm sys}}$.  Transfer matrix positivity then gives
\begin{equation}
 0\le \cC_k(r)-\cC_k(\infty)
 \le A'\e^{-m_{\rm pol}r}
\label{eq:positive_syndrome_decoupling}
\end{equation}
and the vortex source mass obeys
\begin{equation}
 a m_k^{(V)}\ge m_{\rm pol}>0.
\label{eq:polymer_mass_lower_bound}
\end{equation}
These constants are independent of $\kappa$ because the known endpoint fugacity cancels from $\cC_k$.
\end{theorem}

\begin{proof}
Expand $\log\cZ_\beta[B]$ in connected clusters. By \eqref{eq:stokes_syndrome_polymer}, a contractible cluster in the pair background depends only on the endpoint current intersected by its filling. Clusters near a single endpoint contribute to the corresponding one defect free energy and build $\cC_k(\infty)$. A cluster that contributes to the difference at finite separation must either connect the two endpoint neighborhoods or have nontrivial homology. A connecting cluster has diameter at least $r$ and is bounded by the exponentially weighted cluster norm. A homologically nontrivial cluster gives the finite volume term in \eqref{eq:topological_remainder_bound}. This proves \eqref{eq:polymer_syndrome_decoupling}. The positive transfer matrix representation gives \eqref{eq:positive_syndrome_decoupling}, and comparison of decay rates gives \eqref{eq:polymer_mass_lower_bound}.
\end{proof}

For $\SU(2)$, the leading connected cluster is a tube with a one plaquette cross section that joins the two source worldlines. It contains four fundamental plaquettes per lattice step and begins at order $u^{4r}$. The even source has nonzero leading overlap with the strong coupling $A_1^{++}$ state. At sufficiently small $u$, the decoder mass lies on the same isolated transfer matrix branch and has the usual scalar strong coupling expansion \cite{Munster,LangelageEtAl}:
\begin{equation}
 a m_{\rm dec}
 =a m_{A_1^{++}}
 =-4\log u+2u^2-\frac{98}{3}u^4+O(u^6).
\label{eq:decoder_and_glueball_series}
\end{equation}
This equality compares the decoder mass with the glueball effective mass at fixed cutoff in the strong coupling region. Extending it to the continuum scaling region requires estimates that remain uniform along the scaling trajectory.

\section{Thermal center flux sectors}
\label{sec:flat_thermal}

Take $\Lambda=T_L^3\times S^1_{N_t}$ and send $\kappa\to+\infty$. Only the flat sector $J=0$ remains. The K\"unneth decomposition is
\begin{equation}
 H^2(T^3\times S^1;\ZN)
 \cong H^2(T^3;\ZN)
 \oplus\left[H^1(T^3;\ZN)\otimes H^1(S^1;\ZN)\right].
\label{eq:kunneth_temporal_merged}
\end{equation}
The second factor is $(\ZN)^3$ and labels temporal 't Hooft twists. After fixing the purely spatial fluxes, we obtain a subsystem of three logical qudits. Let $\bm{k}\in(\ZN)^3$, and write
\[
 Z_{\bm{k}}(L,T)=\cZ_\beta[B_{\bm{k}}].
\]
Then
\begin{equation}
 \Prob(\bm{k})=\frac{Z_{\bm{k}}}{\sum_{\bm{k}'}Z_{\bm{k}'}}.
\label{eq:thermal_twist_posterior}
\end{equation}

\begin{corollary}[Thermal twist decoding]
\label[corollary]{cor:thermal_twist_decoding}
Suppose the trivial temporal twist is most likely.
\begin{enumerate}[label=(\roman*)]
\item If every nontrivial twist has an interface free energy
\begin{equation}
 -\log\frac{Z_{\bm{k}}}{Z_{\bm{0}}}
 =\tau_{\bm{k}}(T)L^2+o(L^2),
 \qquad \tau_{\bm{k}}(T)>0,
\label{eq:twist_interface_law}
\end{equation}
then
\begin{equation}
 p_{\rm fail}^{\ML}
 \le(N^3-1)\exp[-\tau_{\min}(T)L^2+o(L^2)]\longrightarrow0.
\label{eq:twist_decoding_bound}
\end{equation}
\item If $\max_{\bm{k}}\abs{Z_{\bm{k}}/Z_{\bm{0}}-1}\to0$, then the posterior becomes uniform and
\begin{equation}
 p_{\rm succ}^{\ML}\longrightarrow N^{-3}.
\label{eq:twist_complete_mixing}
\end{equation}
\end{enumerate}
\end{corollary}

\begin{proof}
Apply \eqref{eq:finite_k_failure} to the $N^3$ temporal sectors. In the first case, every nontrivial fugacity is bounded by the smallest interface tension. In the second case, the sector weights are asymptotically equal, and normalization gives the uniform posterior.
\end{proof}

For $\SU(2)$, maximal 't Hooft loops and fixed electric flux ensembles show that temporal twists are suppressed by a center interface free energy in the deconfined phase and become degenerate in the confined phase \cite{deForcrandSmekal,deForcrandNoth}. The magnetic center sheet channel is therefore correctable in the deconfined phase and mixed in the confined phase. This behavior reflects the choice of the magnetic sheet as the error variable. The discrete Fourier transform
\begin{equation}
 Z_e(\bm e)=\frac1{N^3}
 \sum_{\bm k\in(\ZN)^3}\omega^{\bm e\cdot\bm k}Z_{\bm k}
\label{eq:electric_flux_fourier}
\end{equation}
gives the fixed electric flux sectors \cite{tHooftFlux,deForcrandSmekal}. In this basis, confinement makes nonzero electric flux costly. The invariant statement is that the topological defects representing logical failure must be suppressed.

The flat theory gives an exact quantum error correction interpretation of a familiar high energy observable, but it has no nonzero syndromes. At finite $\kappa$, open sheets occur as stochastic errors, and correlations between their endpoints probe the neutral spectrum.

\part{Syndrome spectroscopy and mass gaps}

\section{Logical sectors and spectral gaps}
\label{sec:global_versus_gap}

We now ask whether error correction in a confining theory can imply a mass gap. The first point is that global logical information and local spectral information are different.

\subsection{Neutral excitations}

\begin{proposition}[Global twist data do not imply a full gap]
\label[proposition]{prop:neutral_sector_no_go}
A property that depends only on ratios $\cZ[B]/\cZ[B']$ of center background partition functions cannot imply a mass gap for the full theory unless one adds a completeness assumption.
\end{proposition}

\begin{proof}
Let $\cZ_{\rm conf}[B]$ be a center sensitive theory with the desired twist ratios. Let $\cZ_{\rm crit}$ be a decoupled, center neutral, massless theory. The product
\[
 \cZ_{\rm tot}[B]=\cZ_{\rm conf}[B]\cZ_{\rm crit}
\]
has
\[
 \frac{\cZ_{\rm tot}[B]}{\cZ_{\rm tot}[B']}
 =\frac{\cZ_{\rm conf}[B]}{\cZ_{\rm conf}[B']}.
\]
has exactly the same logical posteriors, twist free energies, and global decoding probabilities as $\cZ_{\rm conf}$. The full theory is nevertheless gapless.
\end{proof}

The counterexample uses a decoupled sector, so it is not meant as a claim about pure Yang--Mills theory. Its role is to isolate the logical obstruction. Center flux observables probe sectors charged under the one form symmetry, while glueballs are center neutral. Any implication from a global threshold to the neutral spectrum must supply a relation between these two kinds of observables.

\subsection{Charged flux sectors}

A logical decoding exponent does give an exact statement about charged sectors. At finite spatial volume, let a positive transfer matrix define the decomposition into fixed electric center flux sectors
\begin{equation}
 \cH=\bigoplus_{e\in\cL}\cH_e,
 \qquad
 Z_e(\tau)=\Tr_{\cH_e}\e^{-\tau H}.
\label{eq:flux_sector_traces}
\end{equation}
We use the normalized $Z_e$ as the posterior for a logical flux variable and assume that the vacuum lies in the sector $e=0$.

\begin{proposition}[Logical exponent equals a charged sector gap]
\label[proposition]{prop:charged_sector_gap}
At fixed spatial volume define
\[
 R(\tau)=\sum_{e\ne0}\frac{Z_e(\tau)}{Z_0(\tau)},
 \qquad
 p_{\rm fail}(\tau)=\frac{R(\tau)}{1+R(\tau)}.
\]
If each sector has finite ground state degeneracy, then
\begin{equation}
 -\lim_{\tau\to\infty}\frac1\tau
 \log\frac{p_{\rm fail}(\tau)}{1-p_{\rm fail}(\tau)}
 =\min_{e\ne0}\bigl(E_{e,0}-E_{0,0}\bigr).
\label{eq:charged_sector_gap}
\end{equation}
In particular, $p_{\rm fail}(\tau)\le C\e^{-\gamma\tau}$ implies $E_{e,0}-E_{0,0}\ge\gamma$ for every $e\ne0$.
\end{proposition}

\begin{proof}
For large $\tau$,
\[
 Z_e(\tau)=g_e\e^{-\tau E_{e,0}}(1+o(1)).
\]
Since $R(\tau)$ is a finite sum of positive exponentials, its logarithmic rate is the smallest energy difference. The identity $p_{\rm fail}/(1-p_{\rm fail})=R$ then gives \eqref{eq:charged_sector_gap}.
\end{proof}

In a confining theory, a wrapped electric flux tube has energy $\sigma L+o(L)$. Equation \eqref{eq:charged_sector_gap} therefore turns a decoding exponent into a gap between one form charge sectors. It says nothing about excitations within $\cH_0$. To probe neutral excitations, we instead study the interaction between local boundary components of an open sheet.

\section{Center monopole syndrome pairs}
\label{sec:syndrome_monopoles}

\subsection{The defect geometry}

Choose a lattice direction $z$ for transfer matrix evolution and a periodic direction $t$ of extent $N_t$. On the dual lattice, place two oppositely oriented loops that wrap the $t$ circle at $z=0$ and $z=r$. Let $\Sigma_{k,r}$ be a dual sheet of center charge $k\in\{1,\dots,N-1\}$ joining the loops. Its Poincar\'e dual on the primal lattice is the plaquette cochain
\[
 B_{k,r}\in C^2(\Lambda;\ZN),
\]
with syndrome
\begin{equation}
 J_{k,r}=\dd B_{k,r}.
\label{eq:pair_syndrome_merged}
\end{equation}
The dual current $\star J_{k,r}$ describes a center monopole and antimonopole pair wrapping Euclidean time. Moving the interior of the sheet changes $B_{k,r}$ by an exact cochain, so \cref{lem:one_form_orbit} leaves the partition function unchanged.

\begin{figure}[t]
\centering
\begin{tikzpicture}[scale=0.91,>=Latex]
\draw[->,thick] (-0.5,0) -- (8.1,0) node[right] {$z$};
\draw[->,thick] (0,-0.4) -- (0,5.4) node[above] {$t$ (periodic)};
\fill[blue!14] (1.4,0.35) rectangle (6.1,4.85);
\draw[deepblue,thick] (1.4,0.35) rectangle (6.1,4.85);
\draw[very thick,red!70!black] (1.4,0.35) -- (1.4,4.85);
\draw[very thick,red!70!black] (6.1,0.35) -- (6.1,4.85);
\draw[->,red!70!black,thick] (1.4,2.0)--(1.4,3.1);
\draw[->,red!70!black,thick] (6.1,3.1)--(6.1,2.0);
\draw[dashed] (-0.15,0.35)--(7.3,0.35);
\draw[dashed] (-0.15,4.85)--(7.3,4.85);
\node[align=center] at (3.75,2.65) {twisted plaquette sheet\\$B_{k,r}$};
\node[below] at (1.4,-0.08) {$0$};
\node[below] at (6.1,-0.08) {$r$};
\draw[<->,thick] (1.4,-0.55)--node[below] {separation $r$}(6.1,-0.55);
\node[align=right,left] at (1.25,3.6) {monopole\\worldline};
\node[align=left,right] at (6.25,3.6) {antimonopole\\worldline};
\node[align=center] at (3.75,5.2) {$t=0\sim N_t$};
\end{tikzpicture}
\caption{A nonzero syndrome used as a spectral probe. The sheet boundary is the measured current $J_{k,r}$ in the code and a heavy center monopole pair in Yang--Mills theory. Evolution in the $z$ direction gives a positive transfer matrix representation.}
\label{fig:monopole_sheet_merged}
\end{figure}
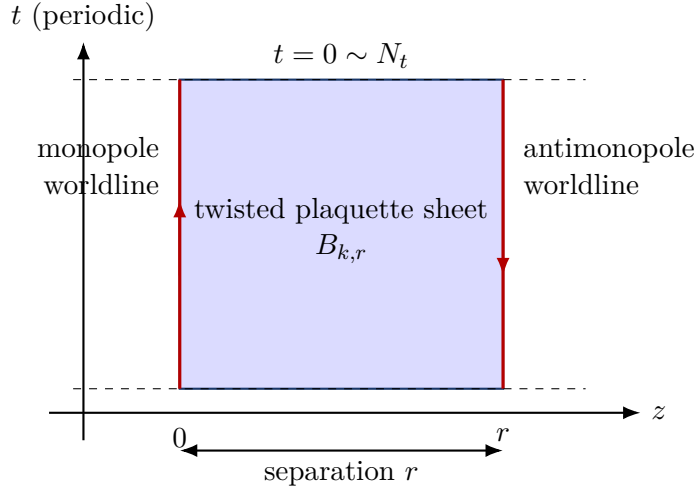

\subsection{Likelihood identity}

Let $W_{\beta,\kappa}(J,h)$ be the joint weight in \eqref{eq:joint_sector_finite_k}. For the syndrome $J_{k,r}$, choose $B_{k,r}$ as the reference completion $h=0$. For $J=0$, choose $B_0=0$.

\begin{theorem}[Syndrome and monopole identity]
\label[theorem]{thm:syndrome_monopole_merged}
Define the core renormalized likelihood for the pair by
\begin{equation}
 \cC_{k,L_z}(r)=
 \frac{\nu_\kappa(0)}{\nu_\kappa(J_{k,r})}
 \frac{W_{\beta,\kappa}(J_{k,r},0)}
 {W_{\beta,\kappa}(0,0)}.
\label{eq:renormalized_pair_likelihood}
\end{equation}
For every finite lattice, finite $\kappa$, and $\beta\ge0$,
\begin{equation}
 \boxed{
 \cC_{k,L_z}(r)=
 \frac{\cZ_\beta[B_{k,r}]}{\cZ_\beta[0]}.}
\label{eq:syndrome_monopole_merged}
\end{equation}
The right hand side is the twisted plaquette center monopole correlator studied in \cite{MonopolePairs}.
\end{theorem}

\begin{proof}
From \eqref{eq:joint_sector_finite_k},
\begin{align*}
 W_{\beta,\kappa}(J_{k,r},0)
 &=\frac{\abs{\im\dd_1}}{\cN_{\beta,\kappa}}
 \nu_\kappa(J_{k,r})\cZ_\beta[B_{k,r}],\\
 W_{\beta,\kappa}(0,0)
 &=\frac{\abs{\im\dd_1}}{\cN_{\beta,\kappa}}
 \nu_\kappa(0)\cZ_\beta[0].
\end{align*}
Taking the ratio cancels the orbit size, the channel normalization, and the known endpoint fugacity.
\end{proof}

For the product fugacity in \eqref{eq:cosine_fugacity}, the two endpoint worldlines have fixed length $N_t$. The ratio $\nu_\kappa(J_{k,r})/\nu_\kappa(0)$ is therefore independent of $r$, and all nontrivial dependence on the separation comes from the $\SU(N)$ field. On a geometry or relative complex with no nontrivial $H^2$ logical completion, \eqref{eq:renormalized_pair_likelihood} can be written directly as a ratio of syndrome probabilities:
\begin{equation}
 \cC_{k,L_z}(r)=
 \frac{\nu_\kappa(0)}{\nu_\kappa(J_{k,r})}
 \frac{\Prob(J_{k,r})}{\Prob(0)}.
\label{eq:open_syndrome_probability}
\end{equation}
On a torus, one must also resolve the global completion, since the same endpoints can be joined by sheets in different $H^2$ classes.

At $\kappa=+\infty$, the pair must be inserted as an external defect because its syndrome has zero probability in the channel. At finite $\kappa$, the pair occurs as a stochastic syndrome event. The normalization in \eqref{eq:renormalized_pair_likelihood} removes the known local production cost and leaves the interaction generated by the Yang--Mills field.

\section{Transfer matrix spectrum}
\label{sec:decoder_spectrum}

For the spectral discussion, take the periodic hypercubic Wilson lattice at theta angle zero, with the reflection plane and sheet geometry used in the standard twisted plaquette construction. The Wilson action is reflection positive and has a positive transfer matrix in the $z$ direction \cite{Luscher}. The twisted sheet is represented by a vortex operator $V_k$ on a transfer slice. The relation between the sheet ratio and this operator was constructed in \cite{MonopolePairs}. Up to the conventional choice between isospectral transfer matrices,
\begin{equation}
 \cC_{k,L_z}(r)=
 \frac{\Tr\!\left(T^{L_z-r}V_kT^rV_k^\dagger\right)}
 {\Tr(T^{L_z})}.
\label{eq:transfer_pair_correlator}
\end{equation}
Averaging over the transverse directions projects to zero transverse momentum.

Taking $L_z\to\infty$ and writing $T=\e^{-aH_z}$ gives
\begin{equation}
 \cC_k(r)=\sum_n
 \abs{\bra nV_k^\dagger\ket0}^2
 \e^{-ar(E_n-E_0)}.
\label{eq:pair_spectral_decomposition}
\end{equation}
Hence
\begin{align}
 \cC_k(\infty)&=\abs{\bra0V_k\ket0}^2,
\label{eq:pair_asymptote}\\
 \cC_k^{\rm conn}(r)
 &:=\cC_k(r)-\cC_k(\infty)
 =\sum_{n>0}c_{n,k}\e^{-ar(E_n-E_0)},
 \quad c_{n,k}\ge0.
\label{eq:pair_connected_spectrum}
\end{align}
Without the transverse momentum projection, a continuum of momenta can produce a power law prefactor. It does not change the exponential decay rate.

\begin{definition}[Decoder correlation length]
\label[definition]{def:decoder_correlation_length}
When the connected likelihood is nonzero for arbitrarily large $r$, define
\begin{equation}
 \xi_{{\rm dec},k}^{-1}
 =-\limsup_{r\to\infty}\frac1r
 \log\cC_k^{\rm conn}(r).
\label{eq:decoder_correlation_length}
\end{equation}
This is the scale on which the core renormalized evidence for two distant syndrome components becomes additive.
\end{definition}

\begin{theorem}[Decoder spectral theorem]
\label[theorem]{thm:decoder_spectral_theorem}
Let $P_0$ denote the vacuum spectral projection and define
\begin{equation}
 m_{V_k}=\inf\set{E-E_0>0:
 \text{the spectral measure of }(1-P_0)V_k^\dagger\ket0
 \text{ has support at }E},
\label{eq:vortex_coupled_mass}
\end{equation}
with the convention that the infimum of the empty set is $+\infty$. If the asymptotic exponential rate exists, then
\begin{equation}
 \boxed{m_{V_k}=\frac{1}{a\xi_{{\rm dec},k}}.}
\label{eq:mass_decoder_equality}
\end{equation}
More generally, if
\begin{equation}
 0\le\cC_k^{\rm conn}(r)\le A\e^{-r/\xi_*}
\label{eq:exponential_syndrome_decoupling}
\end{equation}
for all sufficiently large $r$, then the nonvacuum spectral measure of $V_k^\dagger\ket0$ is supported in
\begin{equation}
 E-E_0\ge\frac{1}{a\xi_*}.
\label{eq:vortex_gap_lower_bound}
\end{equation}
\end{theorem}

\begin{proof}
At finite transverse volume, \eqref{eq:pair_connected_spectrum} is a positive sum of exponentials. Its logarithmic decay rate is the infimum of the support of the positive spectral measure. If a term with $c_{n,k}>0$ satisfied $a(E_n-E_0)<1/\xi_*$, it would eventually decay more slowly than the assumed bound. This is a contradiction. The same argument applies to the positive spectral integral in infinite volume.
\end{proof}

\begin{definition}[Exponential syndrome decoupling]
We say that the finite $\kappa$ decoder has exponential syndrome decoupling for a defect family when the connected core renormalized likelihood of two separated components obeys a bound of the form \eqref{eq:exponential_syndrome_decoupling}, uniformly over the allowed local shapes and decorations.
\end{definition}

Exponential syndrome decoupling is stronger than global correctability. Global decoding asks which topological completion is most likely once the full syndrome is known. Syndrome decoupling instead asks how rapidly the posterior interaction between separated components vanishes. It is a local mixing property of the Bayesian model.

\subsection{The scalar glueball channel}

The positive pair correlator in \eqref{eq:transfer_pair_correlator} is generated by $V_k$ and $V_k^\dagger$. Parity and charge conjugation exchange the vortex source with its adjoint. We may therefore form the even and odd combinations
\begin{equation}
 V_{k,\pm}=V_k\pm V_k^\dagger.
\label{eq:even_vortex_source}
\end{equation}
The transfer matrix analysis of \cite{MonopolePairs} places the magnetic screening mass from this twisted sheet correlator in the finite temperature channel $J_R^{PC}=0_+^{++}$. In the isotropic zero temperature limit, Euclidean rotations identify the relevant reflection and parity quantum numbers, leaving the four dimensional scalar $0^{++}$ channel. For $\SU(2)$, $k=-k$, so the even projection is automatic.

Let $m_{0^{++}}^{(V_k)}$ be the infimum of the nonvacuum $0^{++}$ spectral support in the vortex pair correlator. This notation distinguishes the mass seen by the source from the lightest scalar mass of the theory. A particular source need not overlap with the lightest state in its symmetry channel.

\begin{corollary}[Source coupled scalar gap]
\label[corollary]{cor:conditional_scalar_gap}
Assume the zero temperature limit in which the exact vortex pair correlator has the $0^{++}$ quantum numbers just described. If its connected syndrome likelihood obeys \eqref{eq:exponential_syndrome_decoupling}, then
\begin{equation}
 m_{0^{++}}^{(V_k)}\ge\frac{1}{a\xi_*}>0.
\label{eq:scalar_gap_bound}
\end{equation}
When the asymptotic rate exists, $m_{0^{++}}^{(V_k)}=1/(a\xi_{{\rm dec},k})$. If, in addition, the vortex source has nonzero overlap with the lightest physical $0^{++}$ state, then $m_{0^{++}}^{(V_k)}=m_{0^{++}}$ and the same bound applies to the scalar glueball mass itself.
\end{corollary}

Symmetry determines the channel, but it does not guarantee overlap with the lightest state in that channel. The overlap can be tested numerically by comparing the effective mass from \eqref{eq:pair_connected_spectrum} with a variational basis of conventional glueball operators.

\section{Conditions for a full transfer matrix gap}
\label{sec:full_gap_criterion}

Let $\{V_\alpha\}_{\alpha\in A}$ be a family of defect operators obtained by varying the charge, orientation, shape, transverse representation, and local decoration of the syndrome insertion. Their connected correlators have the form
\[
 \cC_\alpha^{\rm conn}(r)
 =\sum_{n>0}c_{n,\alpha}\e^{-ar(E_n-E_0)}.
\]

\begin{definition}[Spectrally complete syndrome family]
\label[definition]{def:spectrally_complete_family}
Let $P_I$ denote the spectral projection of $H_z$ onto a Borel interval $I\subset(E_0,\infty)$. The family $\{V_\alpha\}$ is \emph{spectrally complete above the vacuum} if, whenever $P_I$ is nonzero, there exists an $\alpha$ such that
\begin{equation}
 P_I V_\alpha^\dagger\ket0\ne0.
\label{eq:spectral_completeness}
\end{equation}
At finite volume, this means that every nonvacuum energy eigenspace receives nonzero spectral weight from at least one source.
\end{definition}

\begin{theorem}[Uniform local correctability implies a full finite volume gap]
\label[theorem]{thm:complete_full_gap}
Suppose $\{V_\alpha\}$ is spectrally complete and that for some $\mu>0$,
\begin{equation}
 0\le\cC_\alpha^{\rm conn}(r)
 \le A_\alpha\e^{-a\mu r}
\label{eq:uniform_family_decay}
\end{equation}
for every $\alpha$ and all sufficiently large $r$. Then the transfer Hamiltonian has
\begin{equation}
 \Delta_{\YM}:=\inf(\operatorname{spec}H_z\setminus\{E_0\})-E_0\ge\mu.
\label{eq:full_finite_lattice_gap}
\end{equation}
For a sequence of transverse volumes, if spectral completeness and the same exponent $\mu$ hold in every volume, then the finite volume transfer gaps are uniformly bounded below by $\mu$.
\end{theorem}

\begin{proof}
Suppose the spectrum intersects $(E_0,E_0+\mu)$. Then there is a nonzero spectral projection $P_I$ in that interval. Spectral completeness gives a source with positive spectral weight in $I$. Its positive Laplace contribution decays more slowly than the bound in \eqref{eq:uniform_family_decay}, which is a contradiction.
\end{proof}

The basic center sheet gives a scalar source. It is not known to generate a spectrally complete family for the full gauge invariant Hilbert space. Such a family would require decorated sheets carrying the lattice quantum numbers that approach all continuum $J^{PC}$ channels. A full gap therefore requires a uniform exponential bound for a complete set of syndrome probes, not merely a global logical threshold.

\subsection{Thermodynamic and continuum limits}

At fixed transverse volume and lattice spacing, \cref{thm:decoder_spectral_theorem,thm:complete_full_gap} are transfer matrix statements. To take the infinite volume and zero temperature limits, the exponential bound must be uniform in the transverse dimensions and the Euclidean time extent. A continuum statement also requires a sequence of bare couplings $\beta\to\infty$, a scale setting $a(\beta)\to0$, construction of the continuum theory, and constants $m_*>0$, $\beta_0$, and $\ell_0$ such that
\begin{equation}
 \frac{1}{a(\beta)\xi_{{\rm dec},\alpha}
 (\beta;L_\perp,N_t)}\ge m_*
\label{eq:continuum_decoder_condition}
\end{equation}
for all $\beta\ge\beta_0$ and all boxes with $L_\perp a(\beta)\ge\ell_0$ and $N_t a(\beta)\ge\ell_0$. For a scalar mass, this bound is needed for the scalar source. For the full gap, it must hold uniformly over a spectrally complete family. A finite correlation length at each fixed cutoff or finite volume is not sufficient. Finite boxes are always spectrally discrete, and the correlation length in lattice units diverges in any massive continuum limit.

Equation \eqref{eq:continuum_decoder_condition} is the estimate needed to pass from a fixed cutoff transfer matrix to the continuum. In this form, the mass gap question becomes a uniform bound on decoder locality in volume and in physical units.

\section{Numerical tests beyond strong coupling}
\label{sec:numerical_continuation}

The joint action is
\begin{equation}
 S_{\beta,\kappa}[U,B]
 =-\frac{\beta}{N}\sum_p
 \operatorname{Re}\tr(\omega^{-B_p}U_p)
 -\kappa\sum_c\cos\left(\frac{2\pi(\dd B)_c}{N}\right).
\label{eq:finite_kappa_action_merged}
\end{equation}
Existing algorithms sample dynamical center flux variables together with $\SU(N)$ links, including flat 't Hooft sectors and their distributions in the confined and deconfined phases \cite{MorikawaSuzuki,MorikawaSuzuki2026}. A finite $\kappa$ simulation can combine link updates, local plaquette updates, global sheet moves, and constrained or umbrella ensembles for rare prescribed syndromes.

\begin{table}[H]
\centering
\caption{Observables in the decoding and Yang--Mills dictionary.}
\label{tab:merged_observables}
\begin{tabularx}{0.96\textwidth}{>{\raggedright\arraybackslash}p{0.24\textwidth} >{\raggedright\arraybackslash}X >{\raggedright\arraybackslash}X}
\toprule
QEC quantity & Lattice estimator & HEP interpretation\\
\midrule
$\Prob(h\mid J)$ & conditional histogram of global sheet class & relative twisted partition functions\\
$\Delta F_J(g)$ & log ratio of conditional sector counts & defect or interface free energy\\
$p_{\rm fail}^{\ML}(J)$ & $1-\max_h\Prob(h\mid J)$ & total fugacity of nontrivial completions\\
$\cC_k(r)$ & core renormalized $W(J_{k,r},0)/W(0,0)$ & center monopole correlator\\
$m_{\rm eff}(r)$ & adjacent separation log ratio of connected $\cC_k$ & vortex channel effective mass\\
\bottomrule
\end{tabularx}
\end{table}

A convenient effective mass estimator is
\begin{equation}
 m_{\rm eff}(r)=\frac1a
 \log\frac{\cC_k^{\rm conn}(r)}
 {\cC_k^{\rm conn}(r+1)}.
\label{eq:decoder_effective_mass}
\end{equation}
This quantity should be compared with conventional scalar glueball operators at the same bare parameters.

A direct numerical study can proceed in four stages.
\begin{enumerate}[label=\textbf{Stage \arabic*:},leftmargin=5.7em]
\item reproduce the flat temporal twist distributions and the finite temperature center transition;
\item lower $\kappa$ and map the averaged failure surface $p_{\rm fail}^{\ML}(\beta,\kappa,N_t)$, separating changes in syndrome density from changes in conditional topological completion;
\item impose the pair current $J_{k,r}$, measure \eqref{eq:renormalized_pair_likelihood}, and compare its asymptotic mass with the conventional $0^{++}$ spectrum;
\item follow $a\xi_{\rm dec}$ along a scale setting trajectory and test whether the decoder correlation length approaches a finite physical limit.
\end{enumerate}

The global threshold surface and the syndrome correlation surface need not coincide. The first measures ambiguity among topological sectors. The second measures the neutral local spectrum.

\section{Discussion}
\label{sec:discussion}

We have described decoding in terms of the weights of topological sectors. This description continues to apply when the logical group grows with system size, when only a subgroup of sectors remains distinguishable, and when matter screens a bare Wilson loop. Fourier dual disorder amplitudes, distinguishability in the environment, and coherent information then give complementary descriptions of the logical posterior.

The two Yang--Mills applications use this structure in different ways. The spatial $\ZN$ code is a fixed point model for the magnetic symmetry algebra expected in a confining adjoint theory, subject to \cref{ass:ym_package}. The Euclidean center sheet model instead defines a correlated Pauli inference problem directly from the Wilson path integral. One form invariance gives an exact posterior at finite volume. The strong coupling expansion separates three contributions. Contractible polymers determine the local syndrome action. Wrapping polymers distinguish global logical sectors. Polymers joining syndrome endpoints determine a transfer matrix mass.

These effects need not change at the same critical surface. In particular, confinement of the underlying gauge field can coexist with maximal mixing of magnetic sheet sectors because the result depends on which center object is used as the error variable. Exact superselection is the limit in which transitions between sectors vanish \cite{BaoSuperselection}. When the sector weights are finite, protection becomes a Bayesian inference problem conditioned on the syndrome.

Several problems remain. The Hamiltonian $\PSU(N)$ construction requires uniform control of the effective center Hamiltonian and its phase. The finite curvature phase boundary and the decoder mass should be continued beyond the convergent strong coupling region. Controlling the full transfer spectrum requires a spectrally complete family of decorated syndrome sources. A circuit level threshold analysis must include the complementary Pauli sector and faulty syndrome extraction. At nonzero theta angle, or for coherent nonabelian fusion errors, one needs an operator valued extension because scalar positive sector weights are no longer sufficient.

The continuum question is local rather than purely topological. Twist free energies determine global logical inference, while syndrome correlations determine neutral spectroscopy. A continuum mass statement requires bounds that are uniform in transverse volume, Euclidean time extent, and physical lattice units along a scaling trajectory. The formulation developed here gives a finite volume observable, a strong coupling region in which it can be controlled, and a direct comparison with conventional glueball operators. The uniform continuum estimate remains a dynamical problem.

\appendix

\section{Orbit counting and relative variants}
\label{app:orbit_counting}

We record the finite chain complex identities used above.

Let
\[
 C_{k+1}\xrightarrow{\partial_{k+1}}C_k\xrightarrow{\partial_k}C_{k-1}
\]
be a chain complex of finite Abelian groups. Fix a realizable syndrome $\sigma$, and choose a reference $e_\sigma$. Then
\[
 \set{e\in C_k:\partial e=\sigma}=e_\sigma+Z_k,
 \qquad Z_k=\ker\partial_k.
\]
Choose one representative $\gamma_h$ in each class of $H_k=Z_k/B_k$. The syndrome fiber then has the disjoint decomposition
\begin{equation}
 e_\sigma+Z_k
 =\bigsqcup_{h\in H_k}\bigl(e_\sigma+\gamma_h+B_k\bigr).
\label{eq:affine_coset_decomposition}
\end{equation}
Equation \eqref{eq:affine_coset_decomposition} proves the posterior formula \eqref{eq:posterior_general}. The map $\partial_{k+1}$ is onto $B_k$, and every fiber is a coset of $\ker\partial_{k+1}$. Therefore
\begin{equation}
 \sum_{a\in C_{k+1}}f(\partial a)
 =\abs{\ker\partial_{k+1}}\sum_{b\in B_k}f(b),
\label{eq:equal_fiber_counting}
\end{equation}
This is \eqref{eq:gauge_partition_general}.

The same proof applies to the relative chain complex $C_\bullet(X,R;G)$. Relative cycles may end on the distinguished boundary $R$, and the logical sectors are $H_k(X,R;G)$. Poincar\'e--Lefschetz duality gives the corresponding absolute or relative cohomology group in the dual description. Relative chains provide the natural language for rough and smooth boundaries, lattice surgery, external probe worldlines, and the open slab version of \eqref{eq:open_syndrome_probability}.

For the finite $\kappa$ cochain model, the same argument uses $\dd$ in place of $\partial$. The factor $\abs{\im\dd_1}$ in \eqref{eq:joint_sector_finite_k} is the size of the orbit under addition of exact sheets. Parametrizing the orbit by all $\lambda\in C^1$ introduces an additional common factor $\abs{\ker\dd_1}$, which cancels from every posterior.

\section{Uniform magnetic phases and finite volume obstructions}
\label{app:uniform_phase}

We give the uniform phase statement used in \cref{sec:psu_code}. Let $\omega=\e^{2\pi\ii/N}$, and consider the translation invariant magnetic term
\begin{equation}
 H_B(s)=-J\sum_\alpha
 \left(\omega^sM_\alpha+\omega^{-s}M_\alpha^\dagger\right),
 \qquad J>0,
\label{eq:uniform_magnetic_phase}
\end{equation}
Assume that every Gauss term is of $X$ type. We will use two cellular realizations.

\begin{description}[leftmargin=3.2em,style=nextline]
\item[Edge picture.] Qudits live on edges and $M_p$ is a plaquette operator of $Z$ type. A depth one product $U_m=\prod_\ell X_\ell^{m_\ell}$ shifts the plaquette phases by the coboundary $\dd m\in C^2(T_L^3;\ZN)$.
\item[Two form picture.] Qudits live on plaquettes and $M_c=\prod_{p\subset\partial c}Z_p^{\epsilon(c,p)}$ is a cube operator. A depth one product $U_m=\prod_pX_p^{m_p}$ shifts the cube phases by $\dd m\in C^3(T_L^3;\ZN)$.
\end{description}

The energy cost of a charge on one term is
\begin{equation}
 \epsilon_B(j)=2J\left[1-\cos\left(\frac{2\pi j}{N}\right)\right],
 \qquad
 \epsilon_B=\min_{j\ne0}\epsilon_B(j).
\label{eq:magnetic_charge_cost}
\end{equation}

\begin{lemma}[Uniform phase lemma]
\label[lemma]{lem:uniform_phase_full}
On the cubical three torus of linear size $L$:
\begin{enumerate}[label=(\roman*)]
\item If $N\mid sL^2$ in the edge picture, or $N\mid sL^3$ in the two form picture, there exists a depth one $U_m$ that commutes with every term of $X$ type and maps $H_B(s)$ to $H_B(0)$. Thus the uniform phase is a gauge choice on those volumes.
\item If $N\nmid sL^3$ in the two form picture, the obstruction is the class
\[
 c_*=sL^3\pmod N\in H^3(T^3;\ZN)\cong\ZN.
\]
All cube terms can be made unfrustrated except for a fixed finite set whose charges sum to $c_*$. The minimum energy penalty is
\begin{equation}
 \min\set{\sum_i\epsilon_B(j_i):
 j_i\ne0,\ \sum_i j_i=c_*\pmod N},
\label{eq:top_form_obstruction_cost}
\end{equation}
which is $O_N(J)$ and independent of $L$.
\item In the edge picture the failed divisibility obstruction is
\[
 (sL^2,sL^2,sL^2)\in H^2(T^3;\ZN)\cong(\ZN)^3.
\]
For $N=2$, $s=1$, and odd $L$, every representative has at least $3L$ frustrated plaquettes, and three straight dual lines attain this bound. The energy shift is therefore $\Theta(JL)$: subextensive, but not $O(1)$.
\end{enumerate}
\end{lemma}

\begin{proof}
Conjugation by $U_m$ multiplies $M_\alpha$ by $\omega^{(\dd m)_\alpha}$ and leaves every $X$ type term unchanged. Removing the uniform phase is equivalent to the cochain equation
\begin{equation}
 \dd m=-s\mathbf 1.
\label{eq:uniform_coboundary_equation}
\end{equation}
In the edge picture, the constant two cochain pairs with each coordinate two torus to give $sL^2$. In the two form picture, the constant three cochain pairs with the fundamental class to give $sL^3$. This proves part (i).

For part (ii), the top coboundary $\dd:C^2\to C^3$ has cokernel $H^3(T^3;\ZN)\cong\ZN$. The only invariant of a cube phase pattern is therefore its total charge. A target pattern with residual charges summing to $c_*$ differs from $s\mathbf1$ by an exact cochain. Minimizing the local clock energy over these residual patterns gives \eqref{eq:top_form_obstruction_cost}.

For part (iii), solvability is controlled by the three periods in $H^2$. When $N=2$ and $L$ is odd, a residual plaquette set representing $(1,1,1)$ must intersect every translate of each coordinate two cycle an odd number of times. There are $L$ translates in each of three orientations, so at least $3L$ plaquettes are required. Three axis aligned dual lines attain this bound.
\end{proof}

A real Schrieffer--Wolff coefficient has phase $+1$ or $-1$. The phase $-1$ equals $\omega^{N/2}$ only when $N$ is even. Thus, on volumes satisfying the two form divisibility condition, the sign of the generated cube term can be removed for even $N$. The lemma does not show that the coefficient is nonzero, does not control the remaining Schrieffer--Wolff terms, and does not determine the sign for odd $N$.

\section{Elementary cube amplitude}
\label{app:cube_amplitude}

The group integral in \eqref{eq:cube_integral} can be evaluated directly. Orient the six plaquettes on the boundary of a cube consistently, and write each fundamental character as a trace of four link matrices. Each of the twelve edges appears once as $U_\ell$ and once as $U_\ell^\dagger$. The identity
\begin{equation}
 \int_{\SU(N)}\dd U\,U_{ij}U^\dagger_{k\ell}
 =\frac1N\delta_{i\ell}\delta_{jk},
\label{eq:haar_two_matrix}
\end{equation}
gives a factor $N^{-12}$ from the link integrals. The remaining contraction has eight independent closed index loops and gives $N^8$. Therefore
\begin{equation}
 \int\prod_{\ell\subset c}\dd U_\ell
 \prod_{p\subset\partial c}\chi_f(U_p)
 =N^{8-12}=\frac1{N^4}.
\label{eq:cube_amplitude_proof}
\end{equation}
This elementary contribution is positive. The complete sixth order coefficient also includes every allowed ordering and intermediate representation, so the calculation does not by itself establish the sign or nonvanishing of the full coefficient.

\section{Finite volume spectral measure}
\label{app:spectral_measure}

Let $T$ be a strictly positive transfer matrix on a finite dimensional or trace class Hilbert space, and let $V$ be bounded. At finite extent $L_z$, define
\[
 C_{L_z}(r)=\frac{\Tr(T^{L_z-r}VT^rV^\dagger)}{\Tr T^{L_z}}.
\]
Write $T\ket n=\lambda_n\ket n$ with $\lambda_0>\lambda_1\ge\cdots\ge0$. Expanding the trace gives
\begin{equation}
 C_{L_z}(r)=
 \frac{\sum_{m,n}\lambda_m^{L_z-r}\lambda_n^r
 \abs{\bra mV\ket n}^2}
 {\sum_j\lambda_j^{L_z}}.
\label{eq:finite_transfer_expansion}
\end{equation}
At fixed $r$, the limit $L_z\to\infty$ projects the index $m$ onto the vacuum and gives \eqref{eq:pair_spectral_decomposition}. If the vacuum is degenerate, one must project to a chosen pure phase or use the corresponding vacuum density matrix. The connected correlator is then obtained by subtracting its large distance limit.

Let $\mu_V$ be the positive spectral measure of $(1-P_0)V^\dagger\ket0$ for the shifted Hamiltonian $H-E_0$. Then
\begin{equation}
 C^{\rm conn}(r)=\int_{(0,\infty)}\e^{-arE}\,\dd\mu_V(E).
\label{eq:laplace_spectral_measure}
\end{equation}
The logarithmic decay rate of a Laplace transform of a positive measure is the infimum of the support of that measure. This proves \eqref{eq:mass_decoder_equality} without assuming a discrete spectrum. Positivity prevents cancellations from hiding a lighter state with nonzero source overlap.

At a generic theta angle, the Euclidean measure is not positive, and the transfer kernel need not define a positive spectral measure. The argument then fails. An extension would require an operator valued or quasiprobability decoder rather than an ordinary stochastic Pauli channel.

\section{\texorpdfstring{Elementary limits of the finite $\kappa$ channel}{Elementary limits of the finite kappa channel}}
\label{app:finite_k_limits}

\subsection{Zero Wilson coupling}

At $\beta=0$, normalized Haar integration gives
\[
 \cZ_0[B]=1
\]
for every $B$. Equation \eqref{eq:finite_k_posterior} therefore gives a uniform posterior over logical completions for every fixed realizable syndrome:
\begin{equation}
 \Prob_{0,\kappa}(h\mid J)=\frac1{\abs{H^2(\Lambda;\ZN)}}.
\label{eq:beta_zero_uniform}
\end{equation}
At $\beta=0$, the Yang--Mills hidden variables therefore carry no information about the topology of the sheet.

\subsection{Flat limit}

For the cosine fugacity, the limit $\kappa\to+\infty$ suppresses every nonzero $J$ relative to $J=0$. After dividing by the common one form gauge volume, summing over the remaining flat fields $B$ gauges the center and gives an untwisted Villain ensemble of $\PSU(N)$ type. The conditional sector probabilities remain the normalized twisted partition functions in \eqref{eq:thermal_twist_posterior}.

\subsection{A uniform center distribution}

If the Yang--Mills weight and the curvature fugacity are both independent of $B$, all affine logical cosets have the same cardinality. The posterior is uniform, the defect free energies vanish, and
\[
 p_{\rm succ}^{\ML}=\abs{H^2}^{-1}.
\]
This is the cochain analogue of the uniform noise point in \cref{prop:wilson_bound}.

\section{Sector estimator}
\label{app:sector_estimator}

Suppose a Markov chain samples the joint measure in \eqref{eq:joint_local_parent}. For a fixed syndrome bin $J$, let $N_{J,h}$ count the sampled configurations with global completion $h$. Ergodicity gives
\begin{equation}
 \widehat\Prob(h\mid J)=\frac{N_{J,h}}{\sum_kN_{J,k}}
 \longrightarrow\Prob(h\mid J).
\label{eq:conditional_histogram_estimator}
\end{equation}
The corresponding maximum likelihood failure estimator is
\begin{equation}
 \widehat p_{\rm fail}^{\ML}(J)
 =1-\max_h\widehat\Prob(h\mid J).
\label{eq:failure_histogram_estimator}
\end{equation}
Direct binning is inefficient for rare pair syndromes. One can instead fix $\dd B=J_{k,r}$ in a constrained ensemble and use thermodynamic integration or umbrella reweighting relative to $J=0$ to estimate \eqref{eq:renormalized_pair_likelihood}. Global sheet updates are necessary to equilibrate $H^2$. Without them, a simulation can appear to find a correctable posterior simply because it has not tunneled between logical sectors.

\section*{Acknowledgments}
\addcontentsline{toc}{section}{Acknowledgments}
Ning Bao acknowledges support from Northeastern University and from the ASCR EXPRESS project Quantum Transforms from Classical Transforms. He also acknowledges useful conversations with Layla Hormozi, Aidan Chatwin-Davies, and Chun Jun Cao. ChatGPT was used in the conduct of this research, in particular for editorial passes and theorem development. All theorems and statements in this paper have been independently verified by the author.


\clearpage
\begin{thebibliography}{99}
\small

\bibitem{DKLP}
E. Dennis, A. Kitaev, A. Landahl, and J. Preskill,
``Topological quantum memory,''
J. Math. Phys. \textbf{43}, 4452 (2002), arXiv:quant-ph/0110143.

\bibitem{WHP}
C. Wang, J. Harrington, and J. Preskill,
``Confinement--Higgs transition in a disordered gauge theory and the accuracy threshold for quantum memory,''
Ann. Phys. \textbf{303}, 31 (2003), arXiv:quant-ph/0207088.

\bibitem{ChubbFlammia}
C. T. Chubb and S. T. Flammia,
``Statistical mechanical models for quantum codes with correlated noise,''
Ann. Inst. H. Poincar\'e D \textbf{8}, 269 (2021), arXiv:1809.10704.

\bibitem{BaoSuperselection}
N. Bao, C. Cao, A. Chatwin-Davies, G. Cheng, and G. Zhu,
``Superselection rules, quantum error correction, and quantum chromodynamics,''
JHEP \textbf{05}, 236 (2025), arXiv:2306.17230.

\bibitem{BaoHolography}
N. Bao, C. Cao, and G. Zhu,
``Deconfinement and error thresholds in holography,''
Phys. Rev. D \textbf{106}, 046009 (2022), arXiv:2202.04710.

\bibitem{RajputRoggeroWiebe}
A. Rajput, A. Roggero, and N. Wiebe,
``Quantum error correction with gauge symmetries,''
npj Quantum Information \textbf{9}, 41 (2023), arXiv:2112.05186.

\bibitem{CarenaEtAl}
M. Carena, H. Lamm, Y.-Y. Li, and W. Liu,
``Quantum error thresholds for gauge-redundant digitizations of lattice field theories,''
Phys. Rev. D \textbf{110}, 054516 (2024), arXiv:2402.16780.

\bibitem{SpagnoliEtAl}
L. Spagnoli, A. Roggero, and N. Wiebe,
``Fault-tolerant simulation of lattice gauge theories with gauge covariant codes,''
Quantum \textbf{10}, 1968 (2026), arXiv:2405.19293.

\bibitem{CarrozzaEtAl}
S. Carrozza, A. Chatwin-Davies, P. A. Hoehn, and F. M. Mele,
``A correspondence between quantum error correcting codes and quantum reference frames,''
arXiv:2412.15317.

\bibitem{LacambraEtAl}
J. P. Lacambra, A. Chatwin-Davies, M. Honda, and P. A. Hoehn,
``Gauss law codes and vacuum codes from lattice gauge theories,''
arXiv:2604.06087.

\bibitem{YaoSU2}
X. Yao,
``Quantum error correction codes for truncated $\SU(2)$ lattice gauge theories,''
Phys. Rev. D \textbf{113}, 114512 (2026), arXiv:2511.13721.

\bibitem{LiODeaKhemani}
Y. Li, N. O'Dea, and V. Khemani,
``Perturbative stability and error-correction thresholds of quantum codes,''
PRX Quantum \textbf{6}, 010327 (2025), arXiv:2406.15757.

\bibitem{LiuXuPollmannKnap}
Y.-J. Liu, W.-T. Xu, F. Pollmann, and M. Knap,
``Information-theoretic principle of emergent 1-form symmetries,''
arXiv:2502.17572.

\bibitem{Xu3DToric}
J.-Z. Xu, Y. Zhong, M. A. Martin-Delgado, H. Song, and K. Liu,
``Phenomenological noise models and optimal thresholds of the 3D toric code,''
Quantum Sci. Technol. \textbf{11}, 035022 (2026), arXiv:2510.20489.

\bibitem{Stahl}
C. Stahl,
``Self-correction from higher-form symmetry protection on a boundary,''
PRX Quantum \textbf{4}, 030341 (2023), arXiv:2206.05294.

\bibitem{KramersWannier}
H. A. Kramers and G. H. Wannier,
``Statistics of the two-dimensional ferromagnet. Part I,''
Phys. Rev. \textbf{60}, 252 (1941).

\bibitem{Wegner}
F. J. Wegner,
``Duality in generalized Ising models and phase transitions without local order parameters,''
J. Math. Phys. \textbf{12}, 2259 (1971).

\bibitem{KadanoffCeva}
L. P. Kadanoff and H. Ceva,
``Determination of an operator algebra for the two-dimensional Ising model,''
Phys. Rev. B \textbf{3}, 3918 (1971).

\bibitem{Nishimori}
H. Nishimori,
``Internal energy, specific heat and correlation function of the bond-random Ising model,''
Prog. Theor. Phys. \textbf{66}, 1169 (1981).

\bibitem{FanEtAl}
R. Fan, Y. Bao, E. Altman, and A. Vishwanath,
``Diagnostics of mixed-state topological order and breakdown of quantum memory,''
PRX Quantum \textbf{5}, 020343 (2024), arXiv:2301.05689.

\bibitem{LessaEtAl}
L. A. Lessa, R. Ma, J.-H. Zhang, Z. Bi, M. Cheng, and C. Wang,
``Strong-to-weak spontaneous symmetry breaking in mixed quantum states,''
PRX Quantum \textbf{6}, 010344 (2025), arXiv:2405.03639.

\bibitem{LeeCoherent}
J. Y. Lee,
``Exact calculations of coherent information for toric codes under decoherence: identifying the fundamental error threshold,''
Phys. Rev. Lett. \textbf{134}, 250601 (2025), arXiv:2402.16937.

\bibitem{Kitaev}
A. Y. Kitaev,
``Fault-tolerant quantum computation by anyons,''
Ann. Phys. \textbf{303}, 2 (2003), arXiv:quant-ph/9707021.

\bibitem{BHM}
S. Bravyi, M. B. Hastings, and S. Michalakis,
``Topological quantum order: stability under local perturbations,''
J. Math. Phys. \textbf{51}, 093512 (2010), arXiv:1001.0344.

\bibitem{BravyiTerhal}
S. Bravyi and B. M. Terhal,
``A no-go theorem for a two-dimensional self-correcting quantum memory based on stabilizer codes,''
New J. Phys. \textbf{11}, 043029 (2009), arXiv:0810.1983.

\bibitem{CastelnovoChamon}
C. Castelnovo and C. Chamon,
``Topological order in a three-dimensional toric code at finite temperature,''
Phys. Rev. B \textbf{78}, 155120 (2008), arXiv:0804.3591.

\bibitem{GKSW}
D. Gaiotto, A. Kapustin, N. Seiberg, and B. Willett,
``Generalized global symmetries,''
JHEP \textbf{02}, 172 (2015), arXiv:1412.5148.

\bibitem{KapustinThorngren}
A. Kapustin and R. Thorngren,
``Higher symmetry and gapped phases of gauge theories,''
arXiv:1309.4721.

\bibitem{Tachikawa}
Y. Tachikawa,
``Magnetic discrete gauge field in the confining vacua and the supersymmetric index,''
JHEP \textbf{03}, 035 (2015), arXiv:1412.2830.

\bibitem{tHooftConfinement}
G. 't Hooft,
``On the phase transition towards permanent quark confinement,''
Nucl. Phys. B \textbf{138}, 1 (1978).

\bibitem{tHooftFlux}
G. 't Hooft,
``A property of electric and magnetic flux in non-Abelian gauge theories,''
Nucl. Phys. B \textbf{153}, 141 (1979).

\bibitem{AbeEtAl}
M. Abe, O. Morikawa, S. Onoda, H. Suzuki, and Y. Tanizaki,
``Topology of $\SU(N)$ lattice gauge theories coupled with $\ZN$ 2-form gauge fields,''
JHEP \textbf{08}, 118 (2023), arXiv:2303.10977.

\bibitem{deForcrandSmekal}
P. de Forcrand and L. von Smekal,
``'t Hooft loops, electric flux sectors and confinement in $SU(2)$ Yang--Mills theory,''
Phys. Rev. D \textbf{66}, 011504 (2002), arXiv:hep-lat/0107018.

\bibitem{deForcrandNoth}
P. de Forcrand and D. Noth,
``Precision lattice calculation of $\SU(2)$ 't Hooft loops,''
Phys. Rev. D \textbf{72}, 114501 (2005), arXiv:hep-lat/0506005.

\bibitem{MonopolePairs}
P. de Forcrand, C. Korthals-Altes, and O. Philipsen,
``Screening of $Z(N)$ monopole pairs in gauge theories,''
Nucl. Phys. B \textbf{742}, 124 (2006), arXiv:hep-ph/0510140.

\bibitem{FradkinShenker}
E. Fradkin and S. H. Shenker,
``Phase diagrams of lattice gauge theories with Higgs fields,''
Phys. Rev. D \textbf{19}, 3682 (1979).

\bibitem{FredenhagenMarcu}
K. Fredenhagen and M. Marcu,
``Confinement criterion for QCD with dynamical quarks,''
Phys. Rev. Lett. \textbf{56}, 223 (1986).

\bibitem{RotorCodes}
C. Vuillot, A. Ciani, and B. M. Terhal,
``Homological quantum rotor codes: logical qubits from torsion,''
Commun. Math. Phys. \textbf{405}, 53 (2024), arXiv:2303.13723.

\bibitem{GoepfertMack}
M. G\"opfert and G. Mack,
``Proof of confinement of static quarks in three-dimensional $U(1)$ lattice gauge theory for all values of the coupling constant,''
Commun. Math. Phys. \textbf{82}, 545 (1982).

\bibitem{Guth}
A. H. Guth,
``Existence proof of a nonconfining phase in four-dimensional $U(1)$ lattice gauge theory,''
Phys. Rev. D \textbf{21}, 2291 (1980).

\bibitem{FrohlichSpencer}
J. Fr\"ohlich and T. Spencer,
``Massless phases and symmetry restoration in Abelian gauge theories and spin systems,''
Commun. Math. Phys. \textbf{83}, 411 (1982).

\bibitem{HallidayPhase}
I. G. Halliday and A. Schwimmer,
``The phase structure of $\SU(N)/\mathbb Z_N$ lattice gauge theories,''
Phys. Lett. B \textbf{101}, 327 (1981).

\bibitem{HallidayMonopoles}
I. G. Halliday and A. Schwimmer,
``$Z_2$ monopoles in lattice gauge theories,''
Phys. Lett. B \textbf{102}, 337 (1981).

\bibitem{deForcrandJahn}
P. de Forcrand and O. Jahn,
``Comparison of $SO(3)$ and $SU(2)$ lattice gauge theory,''
Nucl. Phys. B \textbf{651}, 125 (2003), arXiv:hep-lat/0211004.

\bibitem{Luscher}
M. L\"uscher,
``Construction of a self-adjoint, strictly positive transfer matrix for Euclidean lattice gauge theories,''
Commun. Math. Phys. \textbf{54}, 283 (1977).

\bibitem{OsterwalderSeiler}
K. Osterwalder and E. Seiler,
``Gauge field theories on a lattice,''
Ann. Phys. \textbf{110}, 440 (1978).

\bibitem{Munster}
G. M\"unster,
``Strong coupling expansions for the mass gap in lattice gauge theories,''
Nucl. Phys. B \textbf{190}, 439 (1981).

\bibitem{LangelageEtAl}
J. Langelage, G. M\"unster, and O. Philipsen,
``Strong coupling expansion for finite temperature Yang--Mills theory in the confined phase,''
JHEP \textbf{07}, 036 (2008), arXiv:0805.1163.

\bibitem{LundowMarkstrom2023}
P. H. Lundow and K. Markstr\"om,
``Revising the universality class of the four-dimensional Ising model,''
Nucl. Phys. B \textbf{993}, 116256 (2023), arXiv:2209.05292.

\bibitem{MorikawaSuzuki}
O. Morikawa and H. Suzuki,
``Direct Monte Carlo computation of the 't Hooft partition function,''
Prog. Theor. Exp. Phys. \textbf{2025}, 063B04 (2025), arXiv:2501.07042.

\bibitem{MorikawaSuzuki2026}
O. Morikawa and H. Suzuki,
``Direct numerical simulation of the 't Hooft partition function and (de)confining phases,''
arXiv:2601.20159.

\end{thebibliography}
\end{document}